\documentclass[11pt,a4paper]{amsart}

\RequirePackage[OT1]{fontenc}
\usepackage{amsthm,amsmath}
\usepackage{amstext}
\usepackage{amssymb}
\usepackage{amsfonts}
\usepackage{enumitem}
\PassOptionsToPackage{numbers,sort&compress,square}{natbib}

\RequirePackage[colorlinks,citecolor=blue,urlcolor=blue]{hyperref}

\usepackage{lmodern}

\usepackage{dsfont}
\usepackage[margin=1in]{geometry}
\usepackage{xcolor}
\usepackage{framed}
\usepackage{booktabs}
\usepackage{pdflscape}
\usepackage{afterpage}
\usepackage{capt-of}

\usepackage{color}
\usepackage{rotating}

\usepackage{epstopdf}

\usepackage{algpseudocode,algorithm}

\usepackage{float}
\usepackage{graphicx}

\allowdisplaybreaks[1]
\numberwithin{equation}{section}
\numberwithin{figure}{section}

\theoremstyle{plain}
\newtheorem{theorem}{Theorem}[section]
\newtheorem{remark}{Remark}

\newtheorem{corollary}{Corollary}
\newtheorem{lemma}{Lemma}
\newtheorem{example}{Example}

\newtheorem{thmx}{Theorem}

\newtheorem*{assumption*}{Assumptions}

\newcommand{\R}{\mathds{R}}
\renewcommand{\P}{\mathds{P}}

\DeclareMathOperator{\rE}{\mathds{E}}
\DeclareMathOperator{\rP}{\mathds{P}}

\newcommand{\re}{\mathrm{e}}
\newcommand{\rd}{\,\mathrm{d}}

\newcommand{\rate}[1]{\langle \nabla U(x),#1\rangle}
\newcommand{\partiald}[1]{\frac{\partial}{\partial x_#1   } }

\newcommand{\ZZ}{\mathcal{Z}}
\newcommand{\BB}{\mathcal{B}}

\newcommand{\lref}{\lambda_{\mathrm{ref}}}
\newcommand{\llll}{\bar{\lambda}}
\newcommand{\grad}{\nabla}

\setlist[enumerate,1]{label=(\textsc{\alph*}),ref=(\textsc{\alph*})}
\setlist[enumerate,2]{label=(\roman*),ref=(\textsc{\alph{enumi}})-(\roman*)}

\title[Exponential Ergodicity of the Bouncy Particle Sampler]{Exponential Ergodicity of the Bouncy Particle Sampler}
\author{ George Deligiannidis}
\address{Department of Mathematics, King's College London, UK}
\email{george.deligiannidis@kcl.ac.uk}

\author{Alexandre Bouchard-C\^ot\'e}
\address{Department of Statistics, University of British Columbia, Canada}
\email{bouchard@stat.ubc.ca}

\author{Arnaud Doucet}
\address{Department of Statistics, University of Oxford, UK}
\email{doucet@stats.ox.ac.uk}


\makeatother

\begin{document}
\maketitle

\begin{abstract}
Non-reversible Markov chain Monte Carlo schemes based on piecewise deterministic Markov processes have been recently introduced in applied probability, automatic control, physics and statistics. Although these algorithms demonstrate experimentally good performance and are accordingly increasingly used in a wide range of applications, geometric ergodicity results for such schemes have only been established so far under very restrictive assumptions. We give here verifiable conditions on the target distribution under which the Bouncy Particle Sampler algorithm introduced in \cite{P_dW_12} is geometrically ergodic. This holds whenever the target satisfies a curvature condition and has tails decaying at least as fast as an exponential and at most as fast as a Gaussian distribution. This allows us to provide a central limit theorem for the associated ergodic averages. When the target has tails thinner than a Gaussian distribution, we propose an original modification of this scheme that is geometrically ergodic. For thick-tailed target distributions, such as $t$-distributions, we extend the idea pioneered in \cite{J_G_12} in a random walk Metropolis context. We apply a change of variable to obtain a transformed target satisfying the tail conditions for geometric ergodicity. By sampling the transformed target using the Bouncy Particle Sampler and mapping back the Markov process to the original parameterization, we obtain a geometrically ergodic algorithm.
\end{abstract}

\section{Introduction}
Let $\bar{\pi}(\rd x)$ be a Borel probability measure on $\mathds{R}^d$ admitting a density $\bar{\pi}(x) = \exp\{-U(x)\}/\zeta $ with respect to the Lebesgue measure $\rd x$ where  $U:\mathds{R}^d\mapsto [0,\infty)$ is a potential function with locally Lipschitz second derivatives. We assume that this potential function can be evaluated pointwise while $\zeta$ is intractable. In this context, one can sample from $\bar{\pi}(\rd x)$ and compute expectations with respect to this measure using Markov chain Monte Carlo (MCMC) algorithms. A wide range of MCMC schemes have been proposed over the past 60 years since the introduction of the Metropolis algorithm.

In particular, non-reversible MCMC algorithms based on piecewise deterministic Markov processes \cite{D_84,D_93} have recently emerged in applied probability \cite{B_R_16,DM_P_16,F_G_M_16,M_16}, automatic control \cite{M_H_10,M_H_12}, physics \cite{K_K_16,M_K_K_14,P_dW_12} and statistics \cite{B_F_R_16,BC_D_V_15,S_T_17,V_BC_D_D_17,W_R_17}. These algorithms perform well empirically so they have already found many applications; see, e.g., \cite{DM_P_16,K_K_16,I_K_15,N_H_16}. However, to the best of our knowledge, quantitative convergence rates for this class of MCMC algorithms have only been established under stringent assumptions: \cite{M_H_10} establishes geometric ergodicity of such a scheme but only for targets with exponentially decaying tails, \cite{M_16} obtains sharp results but requires the state-space to be compact, while \cite{B_D_16,B_R_16,F_G_M_16} consider targets on the real line. 
Similar restrictions apply to limit theorems for ergodic averages, where for example in \cite{B_D_16}, a Central Limit Theorem (CLT) has been obtained but this result is restricted to targets on the real line. Establishing exponential ergodicity and a CLT under weaker conditions is of interest theoretically but also practically as it lays the theoretical foundations justifying calibrated confidence intervals around Monte Carlo estimates (for a review, see, e.g.~\cite{J_H_01}).

We focus here on the Bouncy Particle Sampler algorithm (\textsf{BPS}), a piecewise deterministic MCMC scheme proposed in \cite{P_dW_12} and previously studied in \cite{BC_D_V_15,M_16}, as it has been observed to perform empirically very well when compared to other state-of-the-art MCMC algorithms \cite{BC_D_V_15,P_dW_12}. In addition it has recently been shown in \cite{V_BC_D_D_17} that \textsf{BPS} is the scaling limit of the (discrete-time) reflective slice sampling algorithm introduced in \cite{N_03}. In this paper we give conditions on the target distribution $\bar{\pi}$ under which \textsf{BPS} is geometrically ergodic. These conditions hold whenever the target satisfies a curvature condition and has {\textit{``regular tails"}, that is tails decaying at least as fast as an exponential and at most as fast as a Gaussian.

When the target has tails thinner than a Gaussian, we show how a simple modification of the original \textsf{BPS} provides a geometrically ergodic scheme. 
This modified \textsf{BPS} algorithm uses a position-dependent rate of refreshment. This modification is easy to implement.
 
In the presence of thick-tailed targets which do not satisfy these geometric ergodicity assumptions, we follow the approach adopted in \cite{J_G_12} for the random walk Metropolis algorithm. We perform a change-of-variable to obtain a transformed target verifying our conditions. \textsf{BPS} is then used to sample this transformed target. By mapping back this process to the original parameterization, we obtain a geometrically ergodic algorithm.

We henceforth restrict our attention to dimensions $d \geq 2$; for $d=1$ \textsf{BPS} coincides with the \textsf{Zig-Zag} process and the one-dimensional \textsf{Zig-Zag} process has been shown to be geometrically ergodic under reasonable assumptions in \cite{B_R_16}.

The rest of the paper is structured as follows.  Section~\ref{sec:prelim} contains background information on continuous-time Markov processes, exponential ergodicity and \textsf{BPS}. The main results are stated in Section~\ref{sec:main}. Section~\ref{sec:properties} establishes several useful ergodic properties of \textsf{BPS} and of its novel variants proposed here. The proofs of the main results can be found in Section~\ref{sec:proof}. 
\section{Background and notation\label{sec:prelim}}
Let $\{Z_t:t\geq 0\}$ denote a time-homogeneous, continuous-time Markov process on a topological space $(\mathcal{Z}, \mathcal{B}(\mathcal{Z}))$, where $\mathcal{B}(\mathcal{Z})$ is the Borel $\sigma$-field of $\mathcal{Z}$, and denote its transition semigroup with $\{P^t:t\geq 0\}$. For every initial condition $Z_0:=z \in \mathcal{Z}$, the process $\{Z_t:t\geq 0\}$ is defined on a filtered probability space $\left(\Omega, \mathcal{F}, \{\mathcal{F}_t\}, \mathds{P}^{z}\right)$, with $\{\mathcal{F}_t\}$ the natural filtration, such that for any $n>0$, times $0<t_1<t_2<\dots < t_n$  and any $B_1, \dots, B_n\in \mathcal{B}(\mathcal{Z})$ we have
$$\mathds{P}^z\left\{ Z_{t_1} \in B_1, \dots, Z_{t_n} \in B_n \right\}
= \int_{B_1} P^{t_1}(z, \rd z_1)\prod_{j=2}^n \int_{B_j} P^{t_j-t_{j-1}}(z_{j-1}, \rd z_j).$$
We write $\mathds{E}^z$ to denote expectation with respect to $\mathds{P}^z$.

Let $\mathfrak{B}(\mathcal{Z})$ denote the space of bounded measurable functions on $\mathcal{Z}$, which
is a Banach space with respect to the norm $\|f \|_\infty:=\sup_{z\in \mathcal{Z}}|f(z)|$. We also write $\mathcal{M}(\mathcal{Z})$ for the space of $\sigma$-finite, signed measures on $(\mathcal{Z}, \mathcal{B}(\mathcal{Z}))$. 
Given 
a measurable function $V:\mathcal{Z} \to [1,\infty)$, we define a metric on $\mathcal{M}(\mathcal{Z})$ through 
$$\|\mu\|_V := \sup_{|f|\leq V} | \mu(f)|.$$

For $t\geq 0$, we define an operator $P^t: \mathfrak{B}(\mathcal{Z})\to  \mathfrak{B}(\mathcal{Z})$ through $P^t f(z) = \int P^t(z, \rd w) f(w)$. 
We will slightly abuse notation by letting $P^t$ also denote the dual operator acting on $\mathcal{M}(\mathcal{Z})$ through $\mu P^t (A) = \int_{z\in \mathcal{Z}} \mu(\rd z) P^t(z, A)$, for all $A\in \mathcal{B}(\mathcal{Z})$.
 A $\sigma$-finite measure $\pi$ on $\BB(\ZZ)$ is called \textit{invariant} for $\{P^t:t\geq 0\}$, if $\pi P^t = \pi$ for all $t\geq 0$. 

\subsection{Exponential ergodicity of continuous-time processes}

Suppose  that a Borel probability measure $\pi$ is invariant for $\{P^t:t\geq 0\}$.  We are interested in the exponential convergence of the process in the sense of $V$-\textit{uniform ergodicity}: that is there exists a measurable function $V:\mathcal{Z} \to [1,\infty)$ and constants $D<\infty$, $\rho<1$, such that 
\begin{equation}\label{eq:vuniform}
\|P^t(z, \cdot) - \pi(\cdot)\|_V \leq V(z) D \rho^t, \qquad t\geq 0,
\end{equation}
The proof of $V$-uniform ergodicity usually proceeds through the verification of an appropriate \textit{drift condition} which is often expressed  in terms of the \textit{(strong) generator} of the process (see for example \cite[pg. 28]{D_93}). However, in this paper, it will prove useful to focus on the \textit{extended generator} of the Markov process $\{Z_t:t\geq 0\}$ which is defined as follows.
Let $\mathcal{D}(\widetilde{\mathcal{L}})$ denote the set of measurable functions $f:\mathcal{Z}\to \mathds{R}$ for which there exists a measurable function $h:\mathcal{Z}\to \mathds{R}$ such that $t\mapsto h(Z_t)$ is integrable $\mathds{P}^z$-almost surely for each $z\in \mathcal{Z}$ and the process
$$f(Z_t)-f(z) - \int_0^t h(Z_s)\rd s,\qquad t\geq 0,$$
is a local $\mathcal{F}_t$-martingale.
Then we write $h=\widetilde{\mathcal{L}}f$ and we say that $(\widetilde{\mathcal{L}}, \mathcal{D}(\widetilde{\mathcal{L}}))$ is the \textit{extended generator} of the process $\{Z_t:t\geq 0\}$. This is an extension of the usual \textit{strong generator} associated with a Markov process; for more details see \cite{D_93} and references therein.
We will also need the concepts of \textit{irreducibility}, \textit{aperiodicity}, \textit{small sets} and \textit{petite sets} for which we refer the reader to 
\cite{D_M_T_95}.

\subsection{The Bouncy Particle Sampler}

We begin with some additional notation. We will consider $x\in \mathds{R}^d$ as a column vector and 
we will write $|\cdot|$ and $\langle \cdot, \cdot\rangle$ to denote the Euclidean norm and  scalar product in $\mathds{R}^d$ respectively, whereas { $\|A\| = \sup\{|Ax|:|x|=1\}$ will denote the operator norm of the matrix $A\in \mathds{R}^{d\times d}$}.
Let $B(x,\delta):=\{y\in \mathds{R}^d: |x-y|<\delta\}$. For a function $U:\mathds{R}^d \to \mathds{R}$ we write $\nabla U(x)$ and $\Delta U(x)$ for the gradient and the Hessian of $U(\cdot)$ respectively, evaluated at $x$ and we adopt the convention of treating $\nabla U(x)$ as a column vector. For a differentiable map $h:\mathbb{R}^d\to\mathbb{R}^d$ we will write $\nabla h$ for the Jacobian of $h$; that is, letting $h=(h_1, \dots, h_d)^T$, we have $(\nabla h)_{i,j}= \partial_{x_i} h_j$.}
Let us write $\psi$ for the uniform measure on $\mathds{S}^{d-1}:=\{v\in \mathbb{R}^d: |v|=1\}$, and let $\mathcal{Z}:=\mathds{R}^d\times \mathds{S}^{d-1}$, $\mathcal{B}(\mathcal{Z})$ be the Borel $\sigma$-field on $\mathcal{Z}$ and ${\pi}(\rd x, \rd v) := \bar{\pi}(\rd x) \psi(\rd v)$. Since $\pi$  admits $\bar{\pi}$ as a marginal of its invariant measure, we can use this scheme to approximate expectations with respect to $\bar{\pi}$. For $(x,v)\in \mathcal{Z}$, we also define
\begin{equation}\label{eq:bounceoperator}
R(x)v := v - 2 \frac{\langle \nabla U(x), v \rangle}{|\nabla U(x)|^2}\nabla U(x).
\end{equation}
The vector $R(x)v$ can be interpreted as a Newtonian collision on the hyperplane tangent to the gradient of the potential $U$, hence the interpretation of $x$ as a position, and $v$, as a velocity.

\textsf{BPS} defines a ${\pi}$-invariant, non-reversible, piecewise deterministic Markov process
$\{Z_t:t\geq 0\}=\{(X_t, V_t): t\geq 0\}$ taking values in $\mathcal{Z}$.
We introduce here a slightly more general version of \textsf{BPS} than the one discussed in \cite{B_BC_D_D_F_R_J_16,BC_D_V_15,M_16,P_dW_12}. Let
\begin{equation}\label{eq:lambda}
\bar{\lambda}(x,v) := \lref(x) + \lambda(x,v),    \qquad
\lambda(x,v):= \max\{0, \langle\nabla U(x), v \rangle\}=:\langle \nabla U(x), v \rangle_+,
\end{equation}
where the \textit{refreshment rate} $\lref(\cdot):\mathds{R}^d\mapsto (0,\infty)$ is allowed to depend on the location $x$.  Previous versions of \textsf{BPS} restrict attention to the case $\lref(x)=\lref$; the generalisation considered here will prove useful in establishing the geometric ergodicity of this scheme for thin-tailed targets.

Given any initial condition $z\in\mathcal{Z}$, a construction of a path of \textsf{BPS} is given in Algorithm \ref{algo:bps}. Various methods to simulate exactly $\{\tau_k: k\geq 1\}$ are discussed in \cite{B_BC_D_D_F_R_J_16,BC_D_V_15,P_dW_12}.
\begin{algorithm}
	\caption{ : Bouncy Particle Sampler algorithm}\label{algo:bps}
	\begin{algorithmic}[1]
		\State $(X_0, V_0) \gets (x,v)$
		\State $t_0 \gets 0$
		\For{$k = 1, 2, 3, \dots$}
		\State sample inter-event time $\tau_k$, where $\tau_k$ is a positive random variable such that
		$$\rP[\tau_k \geq t]=\exp\left\{-\int_{r=0}^t \bar{\lambda}(X_{t_{k-1}}+rV_{t_{k-1}}, V_{t_{k-1}})\rd r\right\}$$
		\State for $r \in (0, \tau_k)$, $(X_{t_{k-1} + r}, V_{t_{k-1} + r}) \gets (X_{t_{k-1}}+rV_{t_{k-1}}, V_{t_{k-1}} )$
		\State $t_k \gets t_{k-1} + \tau_k$ \Comment{Time of $k$-th event}
		\State $X_{t_k} \gets X_{t_{k-1}}+\tau_k V_{t_{k-1}}$
		\If {$U_k < \lambda(X_{t_k}, V_{t_{k-1}})/\bar{\lambda}(X_{t_k}, V_{t_{k-1}})$, where $U_k \sim \text{Uniform}(0, 1)$}
		\State $V_{t_k} \gets R(X_{t_{k}}) V_{t_{k-1}}$ \Comment{Newtonian collision on the gradient (``bounce'')}
		\Else
		\State $V_{t_k}\sim \psi$  \Comment{Refreshment of the velocity}
		\EndIf
		\EndFor
	\end{algorithmic}
\end{algorithm}
Equivalently, \textsf{BPS} can be defined as the Markov process on $\mathcal{Z}$ with infinitesimal generator defined by
\begin{equation}\label{eq:generator}
\mathcal{L} f(x,v) = \langle \nabla_x f(x,v), v\rangle +
\llll(x,v) \left[Kf\left(x, v\right) -f(x,v) \right],
\end{equation}
for $f \in \mathcal{D}(\mathcal{L})$, the domain of $\mathcal{L}$,
where the transition kernel $K:\mathcal{Z}\times \mathcal{B}(\mathcal{Z})\mapsto [0,1]$ is defined through
\begin{equation}\label{eq:kernelK}
K\left((x,v), (\rd y, \rd w)\right) =
\frac{\lref(x)}{\llll(x,v)} \delta_x(\rd y) \psi(\rd w)
+\frac{\lambda(x,v)}{\llll(x,v)} \delta_x(\rd y) \delta_{R(x)v}(\rd w),
\end{equation}
where as usual for a  measurable function $f:\mathcal{Z}\to \mathds{R}$ we write
$$Kf(z) := \int_{\mathcal{Z}} f(z') K(z,\rd z').$$

For any $\lref(x)=\lref>0$, it has been shown that the \textsf{BPS} is ergodic provided $U$ is continuously differentiable \cite{BC_D_V_15} when the velocities are distributed according to a normal distribution rather than uniformly on the sphere $\mathds{S}^{d-1}$ as we assume here. Restricting velocities to $\mathds{S}^{d-1}$ makes our calculations more tractable without altering the properties of the process too much. In this context, \cite{M_16} considers only compact state spaces but the arguments therein can be adapted to prove ergodicity in the general case.

\section{Main results\label{sec:main}}
\newcommand{\Yes}{\textbf{Yes}}
\newcommand{\No}{\textbf{No}}
\newcommand{\Unk}{}

\afterpage{%
	\clearpage
	\begin{sidewaystable}
		\centering 
		\vskip 450pt
		\begin{tabular}{lcccccc}
			\toprule 
			\\
			& \multicolumn{6}{c}{Target distributions} \\
			\cline{2-7} \\
			&&\multicolumn{5}{c}{Generalized Gaussian distribution, $\pi(x) \propto \exp(-\|x\|^\beta)$} \\
			\cline{3-7} \\
			&&$\beta \in (0, 1)$&$\beta = 1$&$\beta \in (1,2)$&$\beta = 2$&$\beta > 2$ \\
			Sampling methods &$t$-distributions&Thick tails&Exponential&&Gaussian&Thin tails \\
			\midrule
			BPS and extensions (this work)&\Yes&\Yes&\Yes&\Yes&\Yes&\Yes \\
			&Thm. \ref{thm:subgaussian_tdist_i}&Thm. \ref{thm:subgaussian_tdist_ii}&Thm. \ref{thm:subgaussian_B}&Thm. \ref{thm:subgaussian_A}&Thm. \ref{thm:subgaussian_A}&Thm. \ref{thm:supergaussian} \\
			\midrule 
			Metropolis Adjusted Langevin Algorithm (1D)&\No&\No&\Yes&\Yes&\Yes&\No \\
			&\cite{R_T_96}&\cite{R_T_96}&\cite{M_T_93}&\cite{R_T_96}&\cite{R_T_96}&\cite{R_T_96}\\
			&Thm. 4.3&Thm. 4.3&Sec. 16.1.3&Thm. 4.1&Thm. 4.1 &Thm. 4.2 \\
			\midrule
			Random walk Metropolis--Hastings&\No&\No&\Unk&\Unk&\Yes&\Yes \\
			&\cite{J_T_03}&\cite{J_T_03}&&&\cite{R_T_96b} &\cite{R_T_96b} \\
			\midrule
			Hamiltonian Monte Carlo&\Unk&\No&\Yes&\Yes&\Yes&\No \\
			&&\cite{L_B_B_G_16}&\cite{L_B_B_G_16}&\cite{L_B_B_G_16}&\cite{L_B_B_G_16}&\cite{L_B_B_G_16} \\
			&&Cor. 2.3(ii)&Cor. 2.3(i)&Cor. 2.3(i) &Cor. 2.3(i) &Cor. 2.3(ii)  \\
			\midrule
			Johnson and Geyer (Ann. Statist., 2012)&\Yes&\Unk&\Yes&\Yes&\multicolumn{2}{c}{Transformation} \\
			&\cite{J_G_12}&&\cite{J_G_12}&\cite{J_G_12}& \multicolumn{2}{c}{not needed} \\
			&Sec. 3.3&&Thm. 2 and 4&Cor. 1 and 2 && \\
			\bottomrule
		\end{tabular}
		\captionof{table}{\label{table:examples} Summary of geometric ergodicity (or proven lack of) for various sampling methods on generalized Gaussian distributions and t-distributions commonly used in the MCMC convergence literature. These models cover two important challenging situations: roughly, cases where the gradient of the potential becomes negligible in the tails (two leftmost columns) and cases where both the gradient and the Hessian are unbounded (rightmost column).  See references for precise conditions.}
	\end{sidewaystable}
	\clearpage
}

In this paper, we provide sufficient conditions on the target measure $\bar{\pi}$ and the refreshment rate $\lref$ for \textsf{BPS} to be $V$-uniformly ergodic for the following Lyapunov function\footnote{In \cite{M_H_10}, the Lyapunov function $\re^{U(x)/2} \llll(x,v)^{1/2}$ is used to establish the geometric ergodicity of a different piecewise deterministic MCMC scheme for targets with exponential tails but we found this function did not apply to \textsf{BPS}.}
\begin{equation}\label{eq:lyapunovfn}
V(x,v):= \frac{\re^{U(x)/2}}{\llll(x,-v)^{1/2}}.
\end{equation}

Throughout this section, refer to Table~\ref{table:examples} for examples of target distribution with various tail behaviours where each of our Theorems are used to establish exponential ergodicity.

\begin{assumption*}
Let
$U:\mathds{R}^d\to [0,\infty)$ be such that
\begin{align}
&\frac{\partial^2 U(x)}{\partial x_i \partial x_j}  \mbox{ is locally Lipschitz continuous for all $i,j$},\label{eq:c2abs}\tag{A0}\\
&\int_{\mathds{R}^d} \bar{\pi}(\rd x)|\nabla U(x)| <\infty,
\label{eq:integrability}\tag{A1}\\
&\varliminf_{|x|\to\infty} \frac{\re^{U(x)/2}}{\sqrt{|\nabla U(x)|}} > 0\label{eq:growth_condition}\tag{A2},\\
&V\geq c, \qquad \text{for some $c>0$} \label{eq:lower_bound}\tag{A3}.
\end{align}
\end{assumption*}

\begin{remark}
Assumption \eqref{eq:lower_bound} is not restrictive as in view of Assumption~\eqref{eq:growth_condition}, $V\geq c$ may only fail locally  near the origin. Therefore if $V\geq c$ fails inside a compact set $K$, we can always replace $V$ with
$\widetilde{V}=V+\mathds{1}_K \geq 1$.
\end{remark}

\begin{remark}\label{rem:c2abs2}
From the proofs, it will be clear that Theorems~\ref{thm:subgaussian} and \ref{thm:supergaussian} detailed further remain true if we replace Assumption~(\ref{eq:c2abs})
by the following slightly weaker assumption
\begin{equation}\label{eq:c2abs2}\tag{A0'}
\begin{split}
&t\mapsto \langle \nabla U(x+tv),v \rangle \mbox{ is locally Lipshitz for all $(x,v)\in \mathcal{Z}$, and}\\
&\mbox{ (\ref{eq:c2abs}) holds for all $|x|>R$, for some $R>0$.}
\end{split} 
\end{equation}
Although cumbersome, this alternative formulation will become useful in the proof of Theorem~\ref{thm:subgaussian_tdist}.
\end{remark}

Under Assumption (\ref{eq:integrability}), the embedded discrete-time Markov chain $\{\Theta_k:k\geq 0\}:=\{(X_{\tau_{k}},V_{\tau_{k}}):k\geq 0\}$ admits an invariant probability measure; see \cite{C_90} and Lemma \ref{lemma:invariantjumpchain}. The Lyapunov function (\ref{eq:lyapunovfn}) is proportional to the inverse of the square root of the invariant distribution of this embedded discrete-time Markov chain.

\subsection{``Regular" tails}
We now state our first main result. Let
\[
F\left(u,d\right):=\mathbb{E}\left[\frac{\mathbb{I}\left(\vartheta\in\left[0,\pi/2\right]\right)}{\sqrt{1+u\cos\vartheta}}\right],\quad\mathrm{\vartheta\sim}p_{\vartheta}\left(\cdot\right),
\]
where
\[
p_{\vartheta}\left(\theta\right):=\kappa_{d}\left(\sin\theta\right)^{d-2},\quad\kappa_{d}=\left(\int_{0}^{\pi}\left(\sin\theta\right)^{d-2}\mathrm{d}\theta\right)^{-1},\label{eq:densityangle}
\]
is the density of the angle between a fixed unit length vector and a uniformly distributed vector on $\mathds{S}^{d-1}$.
The following Theorem holds.
\begin{theorem}\label{thm:subgaussian}
Suppose that Assumptions~\eqref{eq:c2abs}-\eqref{eq:lower_bound} hold. Let $\lref(x)=\lref$ and suppose that one of the following sets of conditions holds:
\begin{enumerate}
[itemindent=-.2cm,label=\textsc{(\alph*)},ref={\thetheorem\textsc{(\alph*)}}]
\item \label{thm:subgaussian_A}
$\varliminf_{|x| \to \infty} |\nabla U(x)| =\infty$, $\varlimsup_{|x|\to\infty}\|\Delta U(x)\|\leq \alpha_1<\infty$ and $\lref>(2\alpha_1+1)^2$,
\item\label{thm:subgaussian_B}
$\varliminf_{|x| \to \infty} |\nabla U(x)| =2\alpha_2>0$, $\varlimsup_{|x| \to \infty}\|\Delta U(x)\|\leq C<\infty$ and $\lref\leq \alpha_2/c_d$, where $F(c_d, d)\leq 1/4$.
\end{enumerate}
Then {\textsf{BPS}} is $V$-uniformly ergodic.
\end{theorem}
In summary, \textsf{BPS} with a properly chosen constant refreshment rate $\lref>0$ is exponentially ergodic for targets with tails that decay at least as fast as an exponential, and at most as fast as a Gaussian.
In addition the uniform bound on the Hessian imposes some regularity on the curvature of the target.

Theorem~\ref{thm:subgaussian} does not apply to targets with tails thinner than Gaussian or thicker than exponential distributions. 
As summarised in Table \ref{table:examples}, it is also known that \textit{Metropolis adjusted Langevin  algorithm (MALA)}, see \cite[Theorems~4.2 and 4.3]{R_T_96}, and \textit{Hamiltonian Monte Carlo (HMC)}, see \cite[Theorems~5.13 and 5.17]{L_B_B_G_16}, are not geometrically ergodic for such targets. We now turn our attention to these cases. 
\subsection{Thin-tailed targets}
When the gradient grows faster than linearly in the tails any constant refreshment rate will eventually be negligible. It has been shown in \cite{BC_D_V_15} that \textsf{BPS} without refreshment is not ergodic as the process can get stuck forever outside a ball of any radius. In our case, the refreshment rate does not vanish, but an easy back of the envelope calculation shows that refreshment in the tails  will be extremely rare. This will result in long excursions during which the process will not explore the centre of the space. 
 
The above discussion suggests that, when the target is \textit{thin-tailed}, in the sense that the gradient of its potential grows super-linearly in the tails, we need to scale the refreshment rate accordingly in order for it to remain non-negligible in the tails.
The next result makes this intuition more precise.
\begin{theorem}\label{thm:supergaussian}
Suppose that Assumptions~\eqref{eq:c2abs}-\eqref{eq:lower_bound} hold.
Let $\lref>0$ and define for some $\epsilon>0$
$$\lref(x):=\lref+\frac{|\nabla U(x)|}{\max\{1,|x|^\epsilon\}}.$$
Suppose that
$$\varliminf_{|x| \to \infty} {\frac{|\nabla U(x) |}{{|x|}} =\infty}, \quad
\quad \lim_{|x| \to \infty} \frac{\|\Delta U(x)\|}{|\nabla U(x)|}|x|^\epsilon =0.$$
Then \textsf{BPS} is $V$-uniformly ergodic.
\end{theorem}
It is worth noting that although Langevin diffusions can be geometrically ergodic for thin-tailed targets, they  typically cannot be simulated exactly and when discretised require an additional step, such as a Metropolis filter, to sample from the correct target distribution. This results in non-geometrically ergodic algorithms \cite{R_T_96}.
\subsection{Thick-tailed targets} 
For targets with tails thicker than an exponential, that is when the gradient vanishes in the tails, the lack of exponential ergodicity of gradient-based methods such as MALA and HMC, is natural---the vanishing gradient induces random-walk like behaviour in the tails.
This seems to be the main obstruction preventing extension of Theorem~\ref{thm:subgaussian} to thick-tailed distributions. 

However, similarly to \cite{J_G_12}, we can address this by transforming the target to one satisfying the assumptions of either Theorem~\ref{thm:subgaussian}, or Theorem~\ref{thm:supergaussian}. This guarantees that BPS with respect to the transformed target will be geometrically ergodic.
As in \cite{J_G_12} we define the following functions $f^{(i)} : [0, \infty) \to [0, \infty)$ for $i=1,2$:
\begin{equation}\label{eq:fexp}
f^{(1)}(r)
=
\begin{cases}
\re^{br} - \frac{\re}{3}, & r>\frac{1}{b},\\
r^3 \frac{b^3 \re}{6}+ r\frac{b\re}{2}, &r\leq \frac{1}{b},
\end{cases}
\end{equation}
and
\begin{equation}\label{eq:fpoly}
f^{(2)}(r)
=
\begin{cases}
r, & r\leq R,\\
r+(r-R)^p, &r> R,
\end{cases}
\end{equation}
where $R,b>0$ are arbitrary constants.
We also define the \textit{isotropic transformations} $h^{(i)} : \R^d \to \R^d$, given by
\begin{equation}\label{eq:defn_h}
h^{(i)}(x) := 
\begin{cases} 
\displaystyle\frac{f^{(i)}(|x|)x}{|x|}, &\text{for $x\neq 0$,} \\
0, & \text{for $x=0$.}
\end{cases}
\end{equation}

From \cite[Lemma~1]{J_G_12} it follows that for $i=1,2$,  $h=h^{(i)}: \mathds{R}^d\mapsto\mathds{R}^d$ defines a  $C^1$-\textit{diffeomorphism}, that is $h$ is bijective with $h,h^{-1}\in C^1(\mathds{R^d})$. 

Let $h=h^{(i)}$ for some $i\in \{1,2\}$, $X\sim \bar{\pi}$ and $Y=h^{-1}(X)$. Then $Y\in \mathds{R}^d$ is distributed according to the Borel probability measure  $\bar{\pi}_h$,  with density given by
$\bar{\pi}_h(y)=\exp\{-U_h(y)\}/\zeta_h$, where by \cite[equations (6) and (7)]{J_G_12} we have that
\begin{align}
U_h(y) &= U(h(y)) - \log \det(\nabla h(y)),\label{eq:uhdef}\\
\nabla U_h (y) &=  \nabla h(y)\nabla U(h(y)) - \nabla \log \det (\nabla h(y)).\label{eq:gradUh}
\end{align}
Let $\{(Y_t,V_t);t\geq 0\}$ denote the trajectory produced by the BPS algorithm targeting $\pi_h(y,v):=\bar{\pi}_h(y)\psi(v)$ and let 
$V_h$ be defined through \eqref{eq:lyapunovfn}, similarly with $U_h$ in place of $U$.

\begin{theorem}\label{thm:subgaussian_tdist} Let $U$ satisfy Assumption~(\ref{eq:c2abs}). Then we have the following.
\begin{enumerate}[label={\textsc{(\alph*)}}, ref={\thetheorem\textsc{(\alph*)}}]
\item \label{thm:subgaussian_tdist_i}If
	\begin{enumerate}
		\item \label{assumption:tail} $\varlimsup_{|x|\to \infty} |x||\nabla U(x)| <\infty$,
		\item \label{assumption:curvature} $\varlimsup_{|x|\to \infty}|x|^2 \|\Delta U(x)\|<\infty$, and
		\item \label{assumption:directionality} $\varliminf_{|x|\to \infty}\langle x, \nabla U(x) \rangle  > d$,
	\end{enumerate}
	then $U_{h^{(1)}}$, with $h^{(1)}$ defined via \eqref{eq:fexp}, satisfies the assumptions of Theorem~\ref{thm:subgaussian_B}. In addition, the process
	$\{(X_t,V_t):t\geq 0\}$, where $X_t=h^{(1)}(Y_t)$, is $\pi$-invariant and $\widetilde{V}$-uniformly ergodic,
	where $\widetilde{V}=V_{h^{(1)}}\circ 	H^{(1)}$ with $H^{(1)}(x,v):=(h^{(1)}(x),v)$.
\item \label{thm:subgaussian_tdist_ii}If for some $\beta\in (0,1)$ we have
	\begin{enumerate}
		\item \label{assumption:tail2} $\varlimsup_{|x|\to \infty} |x|^{1-\beta}|\nabla U(x)|<\infty$,
				\item \label{assumption:directionality2} $\varliminf_{|x|\to \infty} |x|^{-\beta}\langle x, \nabla U(x) \rangle >0$, and
		\item \label{assumption:curvature2} $\varlimsup_{|x|\to \infty} |x|^{2-\beta}\|\Delta U(x)\|<\infty$,
	\end{enumerate}
	then $U_{h^{(2)}}$, with $h^{(2)}$ defined via \eqref{eq:fpoly} and $p$ such that $\beta p>2$, satisfies the assumptions of Theorem~\ref{thm:supergaussian}. In addition, the process
	$\{(X_t,V_t):t\geq 0\}$, where $X_t=h^{(2)}(Y_t)$, is $\pi$-invariant and $\widetilde{V}$-uniformly ergodic,
	where $\widetilde{V}=V_{h^{(2)}}\circ 	H^{(2)}$ with $H^{(2)}(x,v):=(h^{(2)}(x),v)$.	
\end{enumerate}
\end{theorem}
\begin{example}
\begin{description}
\item[Multivariate $t$-distribution.]
Suppose that $x\in \mathds{R}^d$, for $d\geq 2$, $k>1$,  and let
$$\bar{\pi}(x) \propto \re^{-U(x)} = \left[1+\frac{|x|^2}{k} \right]^{-\frac{k+d}{2}}.$$
It follows that 
\begin{equation*}
\nabla U(x) = \frac{(k+d)}{\left( k + |x|^2 \right)}x,\quad 
\Delta U(x) = \frac{k+d}{k+|x|^2}\mathds{1}_d -2 \frac{(k+d) x x^T}{\left( k+ |x|^2 \right)^2},
\end{equation*}
where $\mathds{1}_d$ is the $d\times d$ identity matrix. Then $U$  satisfies the conditions of Theorem~\ref{thm:subgaussian_tdist_i}.
\item[Generalised Gaussian distribution.]
Let $U(x) = |x|^\beta$ for some $\beta \in (0,1)$. Then $U$ satisfies the conditions of Theorem~\ref{thm:subgaussian_tdist_ii}.
\end{description}
\end{example}
\begin{remark}In the context of Theorem~\ref{thm:subgaussian_tdist_i},
while geometric ergodicity holds for all positive fixed $b$, tuning this parameter may be useful in practice as pointed out by \cite{J_G_12}.
\end{remark}

\subsection{A Central Limit Theorem}
From the above results we obtain the following CLT for the estimator $T^{-1}  \int_0^T g(Z_s)\rd s$ of $\pi(g)$. This estimator can be computed exactly when $g$ is a multivariate polynomial of the components of $z$; see, e.g., \cite[Section 2.4]{BC_D_V_15}.
 \begin{theorem}\label{thm:clt}
 Suppose that any of the conditions of Theorems~\ref{thm:subgaussian} or \ref{thm:supergaussian}  hold. Let $\varepsilon>0$ such that $W:= V^{1-\varepsilon}$,  satisfies $\pi(W^2)<\infty$.
 Then for any $g: \mathcal{Z}\to \mathds{R}$ such that $g^2\leq W$ and for any initial distribution, we have that
$$\frac{1}{\sqrt{T}} S_T\left[g-\pi(g)\right]  \Rightarrow \mathcal{N}(0,\sigma^2_g),$$
 with
 \begin{equation*}
 	S_T[g]     := \int_0^T g(Z_s)\rd s, \qquad
 	\sigma^2_g := 2 \int \hat{g}(z) \left[ g(z) - \pi(g)\right] \pi(\rd z),
 \end{equation*}
 where $\hat{g}$ is the solution of the Poisson equation $g-\pi(g) = -\mathcal{L} \hat{g}$, and satisfies $|\hat{g}| \leq c_0 (1+ W)$ for some constant $c_0$.
 \end{theorem}
 \begin{corollary}\label{cor:clt}
Suppose that the conditions of Theorem~\ref{thm:subgaussian_tdist_i} or Theorem~\ref{thm:subgaussian_tdist_ii} hold, let 
$h=h^{(1)},h^{(2)}$ respectively, define $H(x,v)=(h(x),v)$, and let
$\widetilde{V}$ denote the corresponding Lyapunov function. Let $\varepsilon>0$ such that $W:= \widetilde{V}^{1-\varepsilon}$,  satisfies $\pi_h(W^2)<\infty$.
 Then for any $g: \mathcal{Z}\to \mathds{R}$ such that $g^2\leq W$ and for any initial distribution, we have that
\begin{align*}
\frac{1}{\sqrt{T}} \int_{0}^T \left[ g(X_t,V_t) - \pi(g)\right]\rd t
&=\frac{1}{\sqrt{T}} \int_{0}^T \left[ g\circ H(Y_t,V_t) - \pi_h\left(g\circ H\right)\right]\rd t\Rightarrow \mathcal{N}(0,\widetilde{\sigma}_g^2),
 \end{align*}
with 
 $$\widetilde{\sigma}^2_g:=2 \int \widehat{{g\circ H}}(z) \left[ g\circ H(z) - \pi_h(g)\right] \pi_h(\rd z),$$
where $\widehat{g\circ H}$ is the solution of the Poisson equation $g\circ H-\pi\left(g\circ H\right) 
 = -\mathcal{L}_h \widehat{g\circ H}$, and $\mathcal{L}_h$ is given in \eqref{eq:generator} with $\bar{\lambda}$ defined in \eqref{eq:lambda} with $U$ replaced by $U_h$ and $K$ defined in \eqref{eq:kernelK} using $R(x)v$ defined in \eqref{eq:bounceoperator} with $\nabla U_h$ replacing $\nabla U$. 
 \end{corollary}
\section{Auxiliary results\label{sec:properties}}
To prove $V$-uniform ergodicity we will use the following result.
\begin{thmx}{\cite[Theorem~5.2]{D_M_T_95}}\label{thm:DMT}
Let $\{Z_t:t\geq 0\}$ be a Borel right Markov process taking values in a locally compact, separable metric space $\mathcal{Z}$ and assume it is non-explosive, irreducible and aperiodic.
Let
$(\widetilde{\mathcal{L}}, \mathcal{D}(\widetilde{\mathcal{L}}))$ be its extended generator. Suppose that there exists a measurable function $V:\mathcal{Z}\to [1,\infty)$ such that $V\in \mathcal{D}(\widetilde{\mathcal{L}})$, and that for a petite set $C\in \mathcal{B}(\mathcal{Z})$ and constants $b,c>0$ we have
\begin{equation}
\widetilde{\mathcal{L}} V \leq -c V + b \mathds{1}_C.\label{eq:driftCondition}\tag{$\mathfrak{D}$}
\end{equation}
Then $\{Z_t:t\geq 0\}$ is $V$-uniformly ergodic.
\end{thmx}
The \textsf{BPS} processes considered in this paper can be easily seen to satisfy the \textit{standard conditions} in \cite[Section 24.8]{D_93}, and thus by \cite[Theorem~27.8]{D_93} it follows that they are Borel right Markov processes. In addition since the process moves at unit speed, for any $z=(x,v)\in \mathcal{Z}$ the first exit time from $B(0,|x|+M)\times \mathds{S}^{d-1}$ {is at least $M$,} and thus, \textsf{BPS} is non-explosive. 

We will next show that \textsf{BPS} remains $\pi$-invariant when the refreshment rate is allowed to vary with $x$, and that it is irreducible and aperiodic. Finally we will show that all compact sets are \textit{small}, hence \textit{petite}. To complete the proofs of Theorems~\ref{thm:subgaussian} and \ref{thm:supergaussian} it remains to establish \eqref{eq:driftCondition} which is done in Section~\ref{sec:proof}.

\begin{lemma}\label{lemma:invariantjumpchain}
The \textsf{BPS} process is invariant with respect to $\pi$.
\end{lemma}
\begin{proof}
We prove invariance using the approach developed in \cite{C_90}, see also \cite{C_D_08}, where a link is provided between the invariant measures of $\{Z_t:t\geq 0\}$ and those of the embedded discrete-time Markov chain $\{\Theta_k:k\geq 0\}:=\{(X_{\tau_{k}},V_{\tau_{k}}):k\geq 0\}$.
 The Markov transition kernel of this chain is given for $A\times B\in \mathcal{B}(\mathcal{Z})$
by
\begin{align*}
\mathcal{Q}\left((x,v), A\times B \right)
&= \int_{0}^\infty \exp\{-\int_0^s \bar{\lambda}(x+uv,v)\rd u\}
\llll(x+sv,v) K\left((x+sv, v), A\times B \right) \rd s,
\end{align*}
where $K$ is defined in (\ref{eq:kernelK}). We also define for $A\times B\in \mathcal{B}(\mathcal{Z})$ the measure
\begin{align*}
\mu(A\times B):= \int\bar{\lambda}(x,v) \pi(\rd x, \rd v) K\left((x,v), A\times B \right)
&=\int_{A\times B} \left[ \lref(x) + \lambda(x, R(x)v) \right] \pi(\rd x, \rd v) \\
&=\int_{A\times B} \bar{\lambda}(x, -v) \pi(\rd x, \rd v) ,
\end{align*}
as $\lambda(x, R(x)v) = \lambda(x,-v)$. This measure is finite by the integrability condition (\ref{eq:integrability}).
{We set $\xi := (\mu(\mathcal{Z}))^{-1}$ and $\bar \mu :=  \xi \mu$.}
The measure $\bar \mu$ satisfies
$\bar \mu= \mathcal{T} \pi$, where $\mathcal{T}$ is  operator defined in \cite[Section~3.3]{C_90} mapping invariant measures of $\{Z_t:t\geq 0\}$ to invariant measures of $\{\Theta_k:k\geq 0\}$.
By \cite[Theorem~3]{C_90}, $\mathcal{T}$ is invertible.
Therefore, from \cite[Theorem~2]{C_90}, it suffices to prove the result to show that
$\mu$ is invariant for $\{\Theta_k\}$ which we now establish.

For continuous, bounded $f:\mathcal{Z}\to \mathds{R}$  we have 
\begin{align*}
\lefteqn{\zeta\iint \mu(\rd x, \rd v)\mathcal{Q}\left((x,v), \rd y, \rd w\right) f(y,w)}\\
&=\iint \re^{-U(x)} \rd x \psi(\rd v) \llll(x,-v)  \int_{0}^\infty\exp\{-\int_0^s \bar{\lambda}(x+uv,v)\rd u\} \llll(x+sv,v) Pf (x+sv, v)\rd s\\
&=\int_{s=0}^\infty \rd s\iint \re^{-U(x)} \rd x\psi(\rd v) \llll(x,-v)
 \exp\{-\int_0^s \bar{\lambda}(x+uv,v)\rd u\} \llll(x+sv,v) Pf (x+sv, v)
\intertext{and letting $z=x+sv$}
&=\int_{s=0}^\infty \rd s\iint\rd z\psi(\rd v) \re^{-U(z-sv)}  \llll(z-sv,-v)\exp\{-\int_0^s \bar{\lambda}(z+(u-s)v,v)\rd u\} \llll(z,v) Pf (z, v)\\
&=\int_{s=0}^\infty \rd s\iint\rd z\psi(\rd v) \llll(z-sv,-v) \exp\{-U(z-sv) -\int_0^s \bar{\lambda}(z-wv,v)\rd w\} \llll(z,v) Pf (z, v).
\end{align*}
Since
{
\begin{align*}
U(z)
&=U(z-sv) + \int_{w=0}^s \langle \nabla U(z-wv), v\rangle \rd w\\
&=U(z-sv)+ \int_{w=0}^s \left[\max\{\nabla \langle U(z-wv), v\rangle,0\}  + \min\{\nabla \langle U(z-wv), v\rangle,0\}\right]\rd w,
\end{align*}
}
it follows that
{
\begin{align*}
U(z) + \int_{w=0}^s \max\{\nabla \langle U(z-wv), -v\rangle,0\}\rd w
&= U(z-sv)+ \int_{w=0}^s \max\{\nabla \langle U(z-wv), v\rangle,0\} \rd w.
\end{align*}
}
Therefore
\begin{align*}
\lefteqn{\zeta \iint \mu(\rd x, \rd v)\mathcal{Q}\left((x,v), \rd y, \rd w\right) f(y,w)}\\
&=\int_{s=0}^\infty \rd s\iint\rd z\psi(\rd v)   \llll(z-sv,-v) \exp\{-U(z)-\int_0^s \bar{\lambda}(z-wv,-v)\rd w\} \llll(z,v) Pf (z, v)\\
&=\iint\re^{-U(z)}\rd z\psi(\rd v) \llll(z,v) Pf (z, v) \int_{s=0}^\infty \rd s \llll(z-sv,-v) \exp\{-\int_0^s \bar{\lambda}(z-wv,-v)\rd w\} \\
&=\iint\re^{-U(z)}\rd z\psi(\rd v) \llll(z,v) Pf (z, v)
= \zeta \iint \pi(\rd z, \rd v) \llll(z,v) Pf (z, v)
\\
&= \zeta \iint \mu(\rd z, \rd v) f(z,v),
\end{align*}
proving that $\mu$ is invariant for $\mathcal{Q}$.
\end{proof}

\begin{remark}
The Markov chain $\{\Theta_k:k\geq 0\}$ admits an invariant probability measure proportional to $\bar{\lambda}(x, -v) \pi(\rd x, \rd v)$. It follows from a simple change of measure argument that under ergodicity and integrability conditions one has
\begin{equation}
\frac{{\sum_{k=1}^{n}g\left(X_{\tau_{k}},V_{\tau_{k}}\right)/\bar{\lambda}\left(X_{\tau_{k}},-V_{\tau_{k}}\right)}}{\sum_{k=1}^{n}1/\bar{\lambda}\left(X_{\tau_{k}},-V_{\tau_{k}}\right)} \to \pi(g)\qquad \text{a.s. as $n\to \infty$}.
 \end{equation}
This is an alternative estimator of $\pi(g)$ compared to $T^{-1}  \int_0^T g(Z_s)\rd s$.
\end{remark}

\begin{lemma}\label{lemma:irreducibility}For all $T>0$, $z:=(x_0, v_0) \in B(0, T/6)\times \mathds{S}^{d-1}$, and Borel set
	$A\subseteq B(0,\tfrac{T}{6})\times \mathds{S}^{d-1}$,
	$$\rP^z(Z_T \in A) \geq {C(T, d, \lref)} \int \int_A \psi(\rd v) \rd x,$$
	for some constant $C(T, d, \lref)>0$ depending only on $T, d, \lref$.
	Hence, all compact sets are small. Moreover, the process $\{Z_t: t\geq 0\}$ is irreducible.
\end{lemma}
\begin{proof}
%
The proof is inspired by \cite{M_16}.
Let $f: B(0,T/6)\times \mathds{S}^{d-1} \to [0,\infty)$ be a bounded positive function.
Let $E$ be the event that exactly two events have occurred up to time $T$, and both of them are refreshments.
Then
\begin{align*}
\lefteqn{\rE^z[ f(Z_T)]}\\
&\geq \rE^z[ f(Z_T) \mathbf{1}_E]\\
&= \int_{\mathds{S}^{d-1}}\int_{\mathds{S}^{d-1}} \psi(\rd v_1) \psi(\rd v_2)
\int_{t=0}^T \int_{s=0}^{T-t} \rd s \rd t
\llll(x_0+t v_0,v_0)\\
&\qquad \exp\left\{-\int_{u=0}^t \llll(x_0+uv_0,v_0)\rd u \right\}
\times \frac{\lref(x+tv_0)}{\llll(x+tv_0,v_0)}\\
&\qquad \times \llll(x_0+t v_0 +s v_1,v_1) \exp\left\{-\int_{w=0}^s \llll(x_0+v_0 t + w v_1,v_1)\rd w \right\}
\frac{\lref(x+tv_0+sv_1)}{\llll(x+tv_0+sv_1,v_1)}
\\
&\qquad \times\exp\left\{-\int_{r=0}^{T-s-t} \llll(x_0+v_0 t + s v_1+ r v_2,v_2)\rd r \right\}f\left(x_0+tv_0+sv_1 +(T-s-t)v_2, v_2\right)\\
&= \int_{\mathds{S}^{d-1}} \int_{\mathds{S}^{d-1}}\psi(\rd v_1) \psi(\rd v_2)
\int_{t=0}^T \int_{s=0}^{T-t} \rd s \rd t
\lref(x+tv_0) \exp\left\{-\int_{u=0}^t \llll(x_0+uv_0,v_0)\rd u \right\}\\
&\qquad \times \lref(x+tv_0+sv_1) \exp\left\{-\int_{w=0}^s \llll(x_0+v_0 t + w v_1,v_1)\rd w \right\}
\\
&\qquad \times\exp\left\{-\int_{r=0}^{T-s-t} \llll(x_0+v_0 t + s v_1+ r v_2,v_2)\rd r \right\}
f\left(x_0+tv_0+sv_1 +(T-s-t)v_2, v_2\right)
.
\end{align*}
As the process moves at unit speed and $|x_0|\leq T/6$, it follows that
$\sup_{t\leq T}|X_t|\leq 7T/6$. Let
$$K:= \sup_{\substack{|x|\leq 7T/6 \\ v\in \mathds{S}^{d-1}}} \llll(x,v) < \infty,$$
and recall that $\llll(x,v) \geq \lref>0$.
Therefore
\begin{align*}
\lefteqn{\rE^z[ f(Z_T)]}\\
&\geq \int_{\mathds{S}^{d-1}} \int_{\mathds{S}^{d-1}}\psi(\rd v_1) \psi(\rd v_2)
\int_{s=0}^T \int_{t=0}^{T-s} \rd s \rd t
\lref\lref  \exp\left\{-\int_{u=0}^t K\rd u -\int_{w=0}^s K\rd w -\int_{r=0}^{T-s-t} K\rd r \right\}
 \\
 &{ \qquad \times f\left(x_0+sv_0+tv_1 +(T-s-t)v_2, v_2\right) }\\
&=\int_{\mathds{S}^{d-1}}\int_{\mathds{S}^{d-1}} \psi(\rd v_1) \psi(\rd v_2)
\int_{s=0}^T \int_{t=0}^{T-s} \rd s \rd t
\lref^2\\
&\qquad \times\exp\left\{-Kt - Ks -K(T-s-t)\right\}f\left(x_0+sv_0+tv_1 +(T-s-t)v_2, v_2\right)\\
&=\int_{\mathds{S}^{d-1}}\int_{\mathds{S}^{d-1}} \psi(\rd v_1) \psi(\rd v_2)
 \int_{s=0}^{T} \rd s \int_{t=s}^{T} \rd t
\lref^2 \exp\left\{-KT\right\}f\left(x_0+sv_0+(t-s)v_1 +(T-t)v_2, v_2\right)\\
&=\int_{\mathds{S}^{d-1}}\int_{\mathds{S}^{d-1}} \psi(\rd v_1) \psi(\rd v_2)
 \int_{t=0}^{T} \rd t \int_{s=0}^{t} \rd s
\lref^2 \exp\left\{-KT\right\}f\left(x_0+sv_0+(t-s)v_1 +(T-t)v_2, v_2\right),\\
&\geq\int_{\mathds{S}^{d-1}}\int_{\mathds{S}^{d-1}} \psi(\rd v_1) \psi(\rd v_2)
 \int_{t=5T/6}^{T} \rd t \int_{s=0}^{t} \rd s
\lref^2 \exp\left\{-KT\right\}f\left(x_0+sv_0+(t-s)v_1 +(T-t)v_2, v_2\right)\\
&=\int_{\mathds{S}^{d-1}}\int_{\mathds{S}^{d-1}} \psi(\rd v_1) \psi(\rd v_2)
 \int_{t=5T/6}^{T} \rd t \,\, t \int_{r=0}^{1} \rd r
\lref^2 \exp\left\{-KT\right\}f\left(x_0+rtv_0+(t-rt)v_1 +(T-t)v_2, v_2\right)\\
&\geq \int_{\mathds{S}^{d-1}}\int_{\mathds{S}^{d-1}} \psi(\rd v_1) \psi(\rd v_2)
\frac{5T}{6} \int_{t=5T/6}^{T} \rd t   \int_{r=0}^{1} \rd r
\lref^2 \exp\left\{-KT\right\}f\left(x_0+trv_0+t(1-r)v_1 +(T-t)v_2, v_2\right).
\end{align*}

Fix $t>5T/6$ and $v_2 \in \mathds{S}^{d-1}$ so that
$x':=x_0+(T-t)v_2$ is now fixed. Since $t>5T/6$ it follows that $T-t<T/6$. Since also $|x_0|\leq T/6$ we must have that $|x'|\leq T/3$. Let $x''\in B(0,T/6)$ be arbitrary. Then it follows that $|x'-x''| \leq T/2$, and therefore there exists
$v_\ast\in \mathds{S}^{d-1}$ and $r_\ast \in [0,1]$ such that
$$x'+tr_\ast v_0 +t(1-r_\ast) v_\ast = x''.$$
Then letting $R\sim U[0,1]$ and $V\sim \psi $ be independent, for $\delta$ small enough  we have
\begin{align*}
\lefteqn{\int_{\mathds{S}^{d-1}} \psi(\rd v_1)  \int_{r=0}^{1} \rd r
\mathds{1}_{B(x'',\delta)} \left(x'+rtv_0+t(1-r)v_1\right)}\\
&= \rP\left\{|x'+tR v_0 +t(1-R) V -x'' | \leq \delta\right\}\\
&= \rP\left\{|tR v_0 +t(1-R) V -tr_\ast v_0 -t(1-r_\ast) v_\ast | \leq \delta\right\} \\
&= \rP\left\{|t(R-r_\ast) v_0 +t(V-v_\ast) -t(RV-r_\ast v_\ast)| \leq \delta\right\} \\
&= \rP\left\{|t(R-r_\ast) v_0 +t(V-v_\ast) -t(RV - Rv_\ast + Rv_\ast-r_\ast v_\ast)| \leq \delta\right\}\\
&= \rP\left\{|t(R-r_\ast) v_0 +t(V-v_\ast) -tR(V -v_\ast)-t(R-r_\ast) v_\ast)| \leq \delta\right\}\\
&\geq \rP\left\{{T}|R-r_\ast| + T|V-v_\ast| +T|V -v_\ast|+T|R-r_\ast| \leq \delta\right\} \\
&= \rP\left\{|R-r_\ast| + |V -v_\ast| \leq \frac{\delta}{2T}\right\} \geq \rP\left\{\max\{|R-r_\ast|, |V -v_\ast|\} \leq \frac{\delta}{4T}\right\}\\
&= \rP\left\{|R-r_\ast|\leq \frac{\delta}{4T}\right\} \rP\left\{ |V -v_\ast| \leq \frac{\delta}{4T}\right\}\geq {\frac{\delta}{4T} \times \Big( C_1  \Big( \frac{\delta}{4T} \Big) \Big)^{d-1}}   = {C_2(T) \delta^d,}
\end{align*}
{where $C_1 > 0$ is a constant, and where $C_i(\cdot)$ denotes quantities depending only on the variables in the bracket.}

Therefore for all $t>5T/6$ and $v_2\in \mathds{S}^{d-1}$ there is a {$C_3(T, d) > 0$} such that
\begin{align*}
\int_{\mathds{S}^{d-1}} \psi(\rd v_1)  \int_{r=0}^{1} \rd r
f \left(x_0+rtv_0+t(1-r)v_1 +(T-t)v_2, v_2\right)
\geq {C_3(T, d)} \int_{B(0, T/6)}f \left(x'', v_2\right) \rd x''
\end{align*}
and thus
\begin{align*}
\rE^z[ f(Z_T)]
&\geq {C_4(T, d, \lref)}\int_{\mathds{S}^{d-1}} \psi(\rd v_2) \int_{t=5T/6}^{T} \rd t \int_{x'' \in B(0,\frac{T}{6})}  f\left(x'', v_2\right)\rd x''\\
&\geq {C_5(T, d, \lref)}\int_{\mathds{S}^{d-1}} \int_{x''\in B(0,\frac{T}{6})}  f\left(x'', v\right)\rd x'' \psi(\rd v),
\end{align*}
and since $f$ is generic, we conclude that for all $z=(x'',v) \in B(0,\tfrac{T}{6})\times \mathds{S}^{d-1}$, and any Borel set
$A\subseteq B(0,\tfrac{T}{6})\times \mathds{S}^{d-1}$
$$\rP^z(Z_T \in A) \geq {C_5(T, d, \lref)} \int \int_A \psi(\rd v) \rd x,$$
whence it follows that for any $R>0$ the set $B(0,R)\times \mathds{S}^{d-1}$ is petite.

Given any compact set $U\subset \mathds{R}^d \times \mathds{S}^{d-1}$, we can find $R>0$ such that
$U\subset B(0,R)\times \mathds{S}^{d-1}$, and we can easily conclude using the above that $U$ must also be petite.

Irreducibility follows easily.
\end{proof}

\begin{lemma}
The process $\{Z_t: t\geq 0\}$ is aperiodic.
\end{lemma}

\begin{proof}
%
We show that for some small set $A'$, there exists a $T$ such that $P^t(z,A')>0$ for all $t\geq T$ and $z\in A'$.

Let $A':=B(0, 1)\times \mathds{S}^{d-1}$, $T = 6$, and suppose that $t > T$. By Lemma~\ref{lemma:irreducibility}, for all $z \in B(0, t/6) \times \mathds{S}^{d-1}$ and Borel set $A \subset B(0, t/6) \times \mathds{S}^{d-1}$, we have $$\P^z(Z_t \in A) \ge C(t, d, \lref) \int \int_{A} \psi(\rd v) \rd x,$$
for some $C(t, d, \lref) > 0$. Hence, by picking $A = A'$, we have, since $B(0, 1) \subset B(0, t/6)$, that for all $z \in A'$,
\begin{equation*}\P^z(Z_t \in A') \ge C(t, d, \lref) \int \int_{A'} \psi(\rd v) \rd x > 0.\qedhere
\end{equation*}
\end{proof}
\section{Proofs of main results\label{sec:proof}}
To complete the proofs of Theorems~\ref{thm:subgaussian} and \ref{thm:supergaussian} 
it remains to show that $V:\mathcal{Z}\to [0,\infty)$ defined in \eqref{eq:lyapunovfn}
satisfies \eqref{eq:driftCondition}.

\subsection{Extended Generator of \textsf{BPS}\label{subsec:extended}}
The expression for the generator provided in \eqref{eq:generator} is not well-defined for $V$, which may not be continuously differentiable at the points $(x,v)$ such that $\langle \nabla U(x),v \rangle=0$.
However, $V$ belongs to $\mathcal{D}(\widetilde{\mathcal{L}})$, the domain of $\widetilde{\mathcal{L}}$, the \textit{extended generator} (see \cite[Section~26]{D_93}) of \textsf{BPS} and this suffices for Theorem~\ref{thm:DMT} to apply.

By Assumption~(\ref{eq:c2abs2}), or the stronger Assumption (\ref{eq:c2abs}),
it easily follows that for all $(x,v)$ the function
$t\mapsto V(x+tv, v)$ is locally Lipschitz so it is absolutely continuous \cite[Proposition~11.8]{D_93}.
Therefore by \cite[Theorem~26.14]{D_93}, since there is no \textsl{boundary} (see \cite[Section~24]{D_93}), $V$ is bounded as a function of $v$ and the jump rate $\bar{\lambda}$ is locally bounded, it follows that $V\in \mathcal{D}(\widetilde{\mathcal{L}})$.

However, at least at points $(x,v)$ such that $\rate{v}=0$, $\nabla_x V(x,v)$ does not exist and therefore the expression given in \eqref{eq:generator} will not make sense. At these points we can express the extended generator in an alternative form given by
\begin{equation}\label{eq:ext_generator}
\widetilde{\mathcal{L}} f(x,v) 
= \mathfrak{V}f(x,v) +
\llll(x,v) \left[Kf\left(x, v\right) -f(x,v) \right],
\end{equation}
where
\begin{equation}\label{eq:vector_field}
\mathfrak{V}f(x,v):= \frac{\rd}{\rd t} f(x+tv,v)\Big\vert_{t=0+},
\end{equation}
which coincides with \eqref{eq:generator} for continuously differentiable functions. The fact that this indeed coincides with the extended generator follows from the local Lipschitz property of $t\mapsto V(x+tv,v)$ and the proof of \cite[Theorem~26.14, bottom of page 71]{D_93}. Indeed, for any fixed $z=(x,v)\in \mathcal{Z}$, let $\{T_i\}_{i\geq 1}$ denote the event times of \textsf{BPS} started from $(x,v)$, the paths of which we denote with $\{Z_t:t\geq 0\}$, where $Z_t=(X_t,V_t)$. Then 
\begin{align*}
V\left( Z_{T_i^-}\right)- V\left( Z_{T_{i-1}}\right) 
&=\int_{0}^{T_i-T_{i-1}} \frac{\rd}{\rd s}V\left( X_{T_{i-1}}+sV_{T_{i-1}}, V_{T_{i-1}}\right)\rd s,
\end{align*}
since the local Lipschitz property of $t\mapsto V(x+tv,v)$ also implies it is almost everywhere differentiable and equal to the integral of its derivative (see e.g. \cite[Proposition~11.8]{D_93}). 
Thus for almost every $t$, the left and right derivatives of $V(x+tv,v)$ coincide and thus
\begin{align*}
V\left( Z_{T_i^-}\right)- V\left( Z_{T_{i-1}}\right) 
&=\int_{0}^{T_i-T_{i-1}} \mathfrak{V}V\left( X_{T_{i-1}}+sV_{T_{i-1}}, V_{T_{i-1}}\right)\rd s.
\end{align*}
From this and the proof of the first part of \cite[Theorem~26.14]{D_93} it follows that 
$$V(Z_t)-V(z)-\int_0^t \mathfrak{V}V(Z_s)\rd s,$$ is a local martingale and thus that 
$\widetilde{\mathcal{L}}$ coincides with the extended generator given in \cite[Eq.(26.15)]{D_93}.

From the discussion in \cite[p.~32]{D_93}, it is clear that for $f\in \mathcal{D}(\widetilde{\mathcal{L}})$, the function $\widetilde{\mathcal{L}}f:\mathcal{Z}\to \mathds{R}$ is uniquely defined everywhere except possibly on a set $A$ of \textit{zero potential}, that is
$$\int_0^\infty \mathds{1}_A(Z_s) \rd s = 0, \quad \text{$\mathds{P}^z$ a.s.},\quad \text{for all $z\in \mathcal{Z}$}.$$
%
For the proof of Theorem~\ref{thm:supergaussian}, $\grad_x V(x,v)$ will not be well defined for the set $A:=\{(x,v)\in\mathcal{Z}: |x|=1\}$ which has zero potential, since the linear trajectories of \textsf{BPS} and the countable number of jumps, imply it can intersect this set at most a countable number of times.

\subsection{Lyapunov functions}
\begin{lemma}[Lyapunov function-Constant refreshment]\label{lem:drift1} Let the refreshment rate be constant, i.e., 
$\lref(x):=\lref$. The function $V$ defined in (\ref{eq:lyapunovfn}) belongs to $\mathcal{D}(\widetilde{\mathcal{L}})$. If either of the conditions of Theorem~\ref{thm:subgaussian} holds, $V$ is a Lyapunov function as it satisfies (\ref{eq:driftCondition}).
\end{lemma}

\begin{proof}
That $V\in \mathcal{D}(\widetilde{\mathcal{L}})$ follows from the discussion in Section~\ref{subsec:extended}.
We now establish that $V$ is a Lyapunov function. First we compute $\widetilde{\mathcal{L}}V(x,v)$. Notice that if $\langle \nabla U(x),v\rangle \neq 0$, then by continuity  there will be a neighborhood of $(x,v)$ on which $V(x,v)$ will be differentiable. Therefore at those points $\widetilde{\mathcal{L}}V(x,v)\equiv \mathcal{L}V(x,v)$. 
\subsubsection*{Case $\rate{v}> 0$}
We have
\begin{align*}
\langle \nabla V(x,v), v\rangle
&= \frac{1}{2}V(x,v)\langle \nabla U(x), v\rangle,
\end{align*}
and adding the reflection part we obtain
\begin{align*}
\lefteqn{\langle \nabla V(x,v), v\rangle + \langle \nabla U(x),v\rangle \left[ V(x, R_x v)-V(x,v)\right]}\\
&= \frac{1}{2}V(x,v)\langle \nabla U(x), v\rangle+ \langle \nabla U(x), v\rangle
\left[ \frac{\re^{U(x)/2}}{\sqrt{\lref+\langle \nabla U(x), -v\rangle_+}}
\frac{\sqrt{\lref+\langle \nabla U(x), -v\rangle_+}}{\sqrt{\lref + \langle \nabla U(x), v\rangle}} - V(x,v)\right]\\
&= -\frac{1}{2}V(x,v)\langle \nabla U(x), v\rangle
+ \langle \nabla U(x), v\rangle V(x,v)
\frac{\sqrt{\lref}}{\sqrt{\lref + \langle \nabla U(x), v\rangle}}.
\end{align*}
The refreshment term is given by
\begin{align*}
\lefteqn{\re^{U(x)/2}\lref \int\psi(\rd w) \left[ \frac{1}{\sqrt{\lref+ \langle \nabla U(x), w\rangle_+}} - \frac{1}{\sqrt{\lref}}\right]}\\
&= \re^{U(x)/2}\lref \int\psi(\rd w)  \frac{1}{\sqrt{\lref+ \langle \nabla U(x), w\rangle_+}}
- \lref V(x,v)\\
&= \re^{U(x)/2}\lref \int_{\rate{w}>0}\psi(\rd w)  \frac{1}{\sqrt{\lref+ \langle \nabla U(x), w\rangle_+}}+\re^{U(x)/2}\lref \int_{\rate{w}\leq 0}\psi(\rd w)  \frac{1}{\sqrt{\lref}}
- \lref V(x,v)\\
&= \re^{U(x)/2}\lref \int_{\rate{w}>0}\psi(\rd w)  \frac{1}{\sqrt{\lref+ \langle \nabla U(x), w\rangle_+}}
-\frac{1}{2} \lref V(x,v),
\end{align*}
since  $\psi\{w: \rate{w}>0\}=1/2$.
Thus overall when $\langle \nabla U(x), v\rangle> 0$ we have
\begin{align*}
\widetilde{\mathcal{L}} V(x,v)
&=
\frac{1}{2}V(x,v)\langle \nabla U, v\rangle
+ \langle \nabla U(x),v\rangle \left[ V(x, R_x v)-V(x,v)\right]\\
&\qquad + \re^{U(x)/2}\lref \int_{\mathds{S}^{d-1}}\psi(\rd w) \left[ \frac{1}{\sqrt{\lref+ \langle \nabla U(x), w\rangle_+}}-
\frac{1}{\sqrt{\lref}}\right]
\\
&=-\frac{1}{2}V(x,v)
\Bigg[ \langle \nabla U(x), v\rangle-
2\frac{\langle \nabla U(x), v\rangle\sqrt{\lref}}{\sqrt{\lref + \langle \nabla U(x), v\rangle}} +\lref\\
&\qquad\qquad -2\lref^{3/2} \int_{\rate{w}\geq 0}\psi(\rd w)  \frac{1}{\sqrt{\lref+ \langle \nabla U(x), w\rangle}} \Bigg]\\
&=-\frac{1}{2}V(x,v)
\Bigg[ \langle \nabla U(x), v\rangle-
2\frac{\langle \nabla U(x), v\rangle\sqrt{\lref}}{\sqrt{\lref + \langle \nabla U(x), v\rangle}} +\lref-2\lref^{} \int_{\theta=0}^{\pi/2}\frac{p_{\vartheta}\left(\theta\right) \rd\theta}{\sqrt{1+ \tfrac{|\nabla U(x)|}{\lref}\cos(\theta)}} \Bigg],
\end{align*}
where $p_{\vartheta}\left(\theta\right)$ is given in (\ref{eq:densityangle}).
\subsubsection*{Case $\rate{v}<0$}
In this case
\begin{align*}
\langle \nabla V(x,v), v\rangle = -\frac{1}{2}V(x,v)\left[ \langle \nabla U, -v\rangle-\frac{1}{\lref+ \rate{-v}} \langle v, \Delta U(x) v\rangle\right].
\end{align*}
Since $\rate{v}_+=0$ there is no reflection and thus overall
\begin{align*}
\widetilde{\mathcal{L}}V(x,v)
&= -\frac{1}{2}V(x,v)\left[ \rate{-v}-\frac{1}{\lref+ \rate{-v}} \langle v, \Delta U(x) v\rangle\right]
-\lref V(x,v)
\\
&\qquad + \frac{1}{2}\lref \frac{\re^{U(x)/2}}{\sqrt{\lref}}+ \lref V(x,v)  \int_{\rate{w}\geq 0}\psi(\rd w)  \frac{\sqrt{\lref+\rate{-v}}}{\sqrt{\lref+ \langle \nabla U(x), w\rangle}}\\
&=-\frac{1}{2}V(x,v)\left[ \rate{-v}-\frac{1}{\lref+ \rate{-v}} \langle v, \Delta U(x) v\rangle +2\lref\right]
\\
&\qquad + \frac{1}{2}\lref V(x,v)\frac{\sqrt{\lref+\rate{-v}}}{\sqrt{\lref}}+ \lref V(x,v)  \int_{\rate{w}\geq 0}\psi(\rd w)  \frac{\sqrt{\lref+\rate{-v}}}{\sqrt{\lref+ \langle \nabla U(x), w\rangle}}\\
&=-\frac{1}{2}V(x,v)\Bigg[\rate{-v} +2\lref-\frac{1}{\lref+ \rate{-v}} \langle v, \Delta U(x) v\rangle
-\lref \frac{\sqrt{\lref+\rate{-v}}}{\sqrt{\lref}}\\
&\qquad\qquad\qquad\qquad\qquad\qquad - 2\lref\int_{\rate{w}\geq 0}\psi(\rd w)  \frac{\sqrt{\lref+\rate{-v}}}{\sqrt{\lref+ \langle \nabla U(x), w\rangle}}
\Bigg].
%
\end{align*}
\subsubsection*{Case $\rate{v}=0$}
In this case we compute $\widetilde{\mathcal{L}}V(x,v)$  as
\begin{align}\label{eq:gen_with_dd}
\widetilde{\mathcal{L}} V(x,v)
&= \frac{\rd}{\rd t} V(x+tv, v) \Big\vert_{t=0+} + \lref \left[ \int \psi(\rd w) V(x,w)- V(x,v)\right],
\end{align}
 since the reflection term vanishes. 

We first compute the directional derivative for which we can distinguish two cases. Suppose first that $\langle \Delta U(x)v, -v\rangle >0$. Then we have that for all $t>0$ small enough
$$\langle \nabla U(x+tv),-v\rangle = 0 + t\langle \Delta U(x)v, -v\rangle +o(t) \geq 0.$$
Therefore, since $\rate{v}=0$, in this case we can compute the first term of \eqref{eq:gen_with_dd} as follows 
\begin{align*}
&\frac{\rd}{\rd t} V(x+tv, v) \Big\vert_{t=0+} \\
&= \lim_{t\to 0+} \frac{1}{t}
	\left[ \frac{\exp(U(x+tv)/2)}{\sqrt{\lref+\langle \nabla U(x+tv), -v\rangle_+}}
	- \frac{\exp\left(U(x)/2\right)}{\sqrt{\lref}}\right]\\
&= \lim_{t\to 0+} \frac{1}{t}
	\left[ \frac{\exp(U(x+tv)/2)}{\sqrt{\lref+\langle \nabla U(x+tv), -v\rangle}}
	- \frac{\exp\left(U(x)/2\right)}{\sqrt{\lref}}\right]\\	
&=\lim_{t\to 0+} \frac{1}{t}
	\left[ \frac{\exp(U(x+tv)/2)-\exp\left(U(x)/2\right)}{\sqrt{\lref+\langle \nabla U(x+tv), -v\rangle}}
	+ \exp\left(U(x)/2\right)\left(\frac{1}{\sqrt{\lref+\langle \nabla U(x+tv), -v\rangle}}-\frac{1}{\sqrt{\lref}}\right)\right]\\	
&=0 -\frac{1}{2}	\frac{\exp\left(U(x)/2\right) }{\sqrt{\lref}^3} \langle \Delta U(x)v, -v\rangle
= - \frac{1}{2}V(x,v) \frac{\langle \Delta U(x)v, -v\rangle}{\lref}
\end{align*}  
Now consider the case where $\langle \Delta U(x)v, -v\rangle \leq 0$, then for all $t>0$ small enough
$$\langle \nabla U(x+tv),-v\rangle = 0 + t\langle \Delta U(x)v, -v\rangle +o(t) \leq 0,$$
and therefore 
\begin{align*}
&\frac{\rd}{\rd t} V(x+tv, v) \Big\vert_{t=0+} \\
&= \lim_{t\to 0+} \frac{1}{t}
	\left[ \frac{\exp(U(x+tv)/2)}{\sqrt{\lref+\langle \nabla U(x+tv), -v\rangle_+}}
	- \frac{\exp\left(U(x)/2\right)}{\sqrt{\lref}}\right]\\
&= \lim_{t\to 0+} \frac{1}{t}
	\left[ \frac{\exp(U(x+tv)/2)}{\sqrt{\lref+0}}
	- \frac{\exp\left(U(x)/2\right)}{\sqrt{\lref}}\right]
=0.\end{align*}  
 Overall we have that 
 \begin{align*}
&\frac{\rd}{\rd t} V(x+tv, v) \Big\vert_{t=0+} 
= - \frac{1}{2}V(x,v) \langle \Delta U(x)v, -v\rangle_+.
\end{align*}  
Adding the refreshment term we find that in this case
\begin{align*}
\widetilde{\mathcal{L}}V(x,v)
&=- \frac{1}{2}V(x,v)
\left[ \frac{\langle \Delta U(x)v, -v\rangle_+}{\lref}
+\lref - 2\lref^{} \int_{\theta=0}^{\pi/2}\frac{p_{\vartheta}\left(\theta\right) \rd\theta}{\sqrt{1+ \tfrac{|\nabla U(x)|}{\lref}\cos(\theta)}} 
\right].
\end{align*}
 
Combining the three cases we obtain
\begin{equation}\label{eq:expression}
2\frac{\widetilde{\mathcal{L}}V(x,v)}{V(x,v)}
=
\begin{cases}
 - 
\Bigg[ \frac{\langle \Delta U(x)v, -v\rangle_+}{\lref}
+\lref \\
\qquad - 2\lref^{} \int_{\rate{w}\geq 0}  \frac{\sqrt{\lref}\psi(\rd w)}{\sqrt{\lref+ \langle \nabla U(x), w\rangle}}  
\Bigg] & \rate{v}=0,
\\
-\Big[ \langle \nabla U(x), v\rangle-
2\frac{\langle \nabla U(x), v\rangle\sqrt{\lref}}{\sqrt{\lref + \langle \nabla U(x), v\rangle}} +\lref  \\
\qquad -2\lref^{} \int_{\rate{w}\geq 0}  \frac{\sqrt{\lref}\psi(\rd w)}{\sqrt{\lref+ \langle \nabla U(x), w\rangle}} \Big]
,&\rate{v}> 0\\
-\Bigg[\rate{-v} +2\lref-\frac{1}{\lref+ \rate{-v}} \langle v, \Delta U(x) v\rangle
\\
\quad -\lref \frac{\sqrt{\lref+\rate{-v}}}{\sqrt{\lref}}\\
\qquad - 2\lref\int_{\rate{w}\geq 0}\psi(\rd w)  \frac{\sqrt{\lref+\rate{-v}}}{\sqrt{\lref+ \langle \nabla U(x), w\rangle}}
\Bigg]
,& \rate{v}<0.
\end{cases}
\end{equation}
\subsection*{Condition (\textsc{a})}
We have that $\varlimsup_{|x|\to\infty}\|\Delta U(x)\|\leq \alpha_1$ and $\varliminf_{|x| \to \infty} |\nabla U(x)| =\infty$.
Thus, since $\lref>0$ and
$1/\sqrt{\cos(\theta)}\in L^1([0,\pi/2], \rd \theta)$,
for any $\epsilon>0$ we can find $K>0$ such that
for all $|x|>K$
\begin{align}
\int_{\rate{w}\geq 0}  \frac{\psi(\rd w)}{\sqrt{\lref+ \langle \nabla U(x), w\rangle}}
&\leq \int_{\theta=0}^{\pi/2}  \frac{p_{\vartheta}\left(\theta\right) \rd \theta}{\sqrt{| \nabla U(x)|\cos(\theta)}}
\leq \frac{\epsilon}{\sqrt{\lref}}\label{eq:ref_term_bound}.
\end{align}
\subsubsection*{Case $\rate{v}=0$} Suppose  that $|x|>K$. Then from \eqref{eq:expression}, by dropping the first term which is negative, 
\begin{align*}
2\frac{\widetilde{\mathcal{L}}V(x,v)}{V(x,v)}
&\leq -
\left[ \lref - 2\lref^{} \int_{\theta=0}^{\pi/2}\frac{\sqrt{\lref}p_{\vartheta}\left(\theta\right) \rd\theta}{\sqrt{\lref+ |\nabla U(x)|\cos(\theta)}} 
\right]\\
&\leq - \lref(1-2\epsilon).
\end{align*}
\medskip
\subsubsection*{Case $\rate{v}>0$}Again let $|x|>K$. From \eqref{eq:expression}
\begin{align*}
2\frac{\widetilde{\mathcal{L}}V(x,v)}{V(x,v)}
&\leq
-
\left[ \langle \nabla U(x), v\rangle-
2\frac{\langle \nabla U(x), v\rangle\sqrt{\lref}}{\sqrt{\lref + \langle \nabla U(x), v\rangle}} +\lref(1 -2\epsilon) \right].
\end{align*}
For $w>0$ consider the function
$$w-2\frac{w\sqrt{\lref}}{\sqrt{\lref+ w}}+\lref(1-2\epsilon).$$
Since for $w>0$, $\lref+w\geq 2\sqrt{w}\sqrt{\lref}$ we have that
\begin{align}
w-2\frac{w\sqrt{\lref}}{\sqrt{\lref+ w}}+\lref(1-2\epsilon)
&\geq w-2\frac{w\sqrt{\lref}}{\sqrt{2}\lref^{1/4}w^{1/4}}+\lref(1-2\epsilon)\notag\\
&=w-\sqrt{2}w^{3/4}\lref^{1/4}+\lref(1-2\epsilon)=:f_\epsilon(w)=:f(w).
\label{eq:fminim}
\end{align}
Then
$$f'(w) = 1- \frac{3\lref^{1/4}}{2\sqrt{2}w^{1/4}},$$
and thus $f$ is minimised at $w_\ast=81\lref/64$ and
$$f(w_\ast)= \left( \frac{37}{64}-2\epsilon\right)\lref.$$
{For any $\lref>0$ we can choose $\epsilon>0$ small enough so that $f(w_\ast)>0$.
From \eqref{eq:ref_term_bound}
we can choose $K$ large enough, so that for all
$|x|>K$ and all $v$ such that $\rate{v}> 0$
$$2\frac{\widetilde{\mathcal{L}}V(x,v)}{V(x,v)}< -\delta,$$
for some $\delta =f(w_\ast)>0$.}
\smallskip
\subsubsection*{Case $\rate{v}<0$}
Then from \eqref{eq:expression}
\begin{align*}
2\frac{\widetilde{\mathcal{L}}V(x,v)}{V(x,v)}
&=
-\Bigg[\rate{-v} +2\lref-\frac{1}{\lref+ \rate{-v}} \langle v, \Delta U(x) v\rangle
-\lref \frac{\sqrt{\lref+\rate{-v}}}{\sqrt{\lref}}\\
&\qquad \qquad - 2\lref\int_{\rate{w}\geq 0}\psi(\rd w)  \frac{\sqrt{\lref+\rate{-v}}}{\sqrt{\lref+ \langle \nabla U(x), w\rangle}}
\Bigg],
\end{align*}
and arguing in the same way as in the previous case, given $\epsilon>0$ we can choose $K>0$ such that for all $|x|>K$ we have similarly to \eqref{eq:ref_term_bound}
\begin{align*}
\int_{\rate{w}\geq 0}\psi(\rd w)  \frac{1}{\sqrt{\lref+ \langle \nabla U(x), w\rangle}}
&\leq \int_{\rate{w}\geq 0}\psi(\rd w)  \frac{1}{\sqrt{\langle \nabla U(x), w\rangle}}\\
&=  \int_{\theta=0}^{\pi/2}  \frac{p_{\vartheta}\left(\theta\right) \rd \theta}{\sqrt{| \nabla U(x)|\cos(\theta)}} \leq
 \frac{\epsilon}{\lref}.
\end{align*}
Since $\varlimsup_{|x|\to\infty} \|\Delta U(x)\| \leq \alpha_1$,
for $K$ large enough and $|x|>K$  we have $\|\Delta U(x)\| \leq 2\alpha_1$. Thus overall when  $\rate{v}<0$
\begin{align*}
\frac{\widetilde{\mathcal{L}}V(x,v)}{V(x,v)}
&\leq
-\frac{1}{2} \Bigg[\rate{-v} +2\lref-\frac{2}{\lref+ \rate{-v}}\alpha_1 \\
&\qquad -\lref \frac{\sqrt{\lref+\rate{-v}}}{\sqrt{\lref}} - 2\epsilon\sqrt{\lref+\rate{-v}}\Bigg].
\end{align*}
For $w=\rate{-v}>0$ define
\begin{align*}
g(w)
&:=w+2\lref-\frac{2\alpha_1}{\lref+ w}
-\lref \frac{\sqrt{\lref+w}}{\sqrt{\lref}} - 2\epsilon\sqrt{\lref+w},\\
g'(w)
&= 1+\frac{2 \alpha_1 }{(\lref +w)^2}-\frac{\sqrt{\lref }}{2 \sqrt{\lref +w}}-\frac{\epsilon  }{\sqrt{\lref +w}}\\
&\geq 1+ \frac{2\alpha_1}{(w+\lref)^2}-\frac{1}{2}-\frac{\epsilon}{ \sqrt{\lref}}
=\frac{1}{2}+\frac{2\alpha_1}{(w+\lref)^2}-\frac{\epsilon}{ \sqrt{\lref}},
\end{align*}
and thus for all $\lref$ we can choose $\epsilon$ small enough so that
$g'(w)\geq 0$ for all $w\geq 0$.
Therefore
$$g(w)\geq g(0)=\lref-2\epsilon \sqrt{\lref}-\frac{2\alpha_1}{\lref}.$$
If $\lref \geq (2\alpha_1+1)^2$ then for $\epsilon$ small enough we have that
$g(w)\geq \delta>0$, for some $\delta$.

Thus, there exists $K>0$ large enough so that for all $|x|>K$ and $v$ such that $\rate{v}<0$ we have $\widetilde{\mathcal{L}}V(x,v)/V\leq -\delta$.
Therefore (\ref{eq:driftCondition}) holds with $C= B(0,K\vee K')\times \mathds{S}^{d-1}$.
\subsection*{Condition (\textsc{b})} Recall that $2\alpha_2:=\varliminf_{|x|\to\infty} |\nabla U(x)|$, so that we can choose $K$ large enough so that for all $|x|>K$ we have $|\nabla U(x)|\geq \alpha_2$.
Thus when $|x|>K$ 
\begin{align*}
\int_{\rate{w}\geq 0}  \frac{\psi(\rd w)}{\sqrt{\lref+ \langle \nabla U(x), w\rangle}}
&=
\int_{\theta=0}^{\pi/2}  \frac{p_{\vartheta}(\rd \theta)}{\sqrt{\lref+ |\nabla U(x)| \cos(\theta)}}\\
&\leq
\frac{1}{\sqrt{\lref}}\int_{\theta=0}^{\pi/2}  \frac{p_{\vartheta}(\rd \theta)}{\sqrt{1+ \frac{\alpha_2}{\lref} \cos(\theta)}}
=:\frac{1}{\sqrt{\lref}} F\left( \frac{\alpha_2}{\lref}, d\right).
\end{align*}
Clearly $F(u, d)\leq F(0,d)= 1/2$ for all $u$ and for all $d$, we have that $F(u,d)\to 0$ as $u\to \infty$.

\subsubsection*{Case $\rate{v}=0$}
For $|x|>K$ and \eqref{eq:expression} we have
\begin{align*}
2\frac{\widetilde{\mathcal{L}}V(x,v)}{V(x,v)}
&\leq -
\left[ \lref - 2\lref^{} \int_{\theta=0}^{\pi/2}\frac{\sqrt{\lref}p_{\vartheta}\left(\theta\right) \rd\theta}{\sqrt{\lref+ |\nabla U(x)|\cos(\theta)}} 
\right]\\
&\leq - \lref\left(1-2F\left(\alpha_2/\lref,d\right)\right)\\
&\leq - \lref\left(1- 2\frac{1}{4}\right)=-\frac{\lref}{2},
\end{align*}
as long as $\alpha_2/\lref > c_d$, with $c_d$ defined as in the statement of Theorem~\ref{thm:subgaussian_B}.

\subsubsection*{Case $\rate{v}>0$}
From \eqref{eq:expression} we have
\begin{equation*}
2\frac{\widetilde{\mathcal{L}}V(x,v)}{V(x,v)}
\leq
-
\left[ \langle \nabla U(x), v\rangle-
2\frac{\langle \nabla U(x), v\rangle\sqrt{\lref}}{\sqrt{\lref + \langle \nabla U(x), v\rangle}} +\lref -2\lref^{} F\left( \frac{\alpha_2}{\lref}, d\right) \right].
\end{equation*}
For $w\geq0$, using again that $\lref+w\geq 2\sqrt{w}\sqrt{\lref}$, we have
$$w-2 \frac{w \sqrt{\lref}}{\sqrt{\lref+w}}+\lref (1-2\epsilon) \geq  w-\sqrt{2}\lref^{1/4}w^{3/4}+\lref(1-2\epsilon)= f_{\epsilon}(w).$$
{Recall from \eqref{eq:fminim} that} for all $w\geq0$
$$f_{\epsilon}(w){\geq}\lref \left(\frac{37}{64}-2\epsilon\right)>0,$$
as long as $\epsilon < 37/128$. For each $d$ and $\alpha_2>0$ we can choose $\lref$ small enough so that $F(\alpha_2/\lref,d)<37/128$.
Then following a similar reasoning as before it can be easily seen that, as long as $\lref$ is small enough, then  there exists a $\delta>0$, such that for all $|x|>K$ and $\rate{v}>0$ we have
\begin{equation*}
2\frac{\widetilde{\mathcal{L}}V(x,v)}{V(x,v)}
\leq -\delta.
\end{equation*}

\subsubsection*{Case $\rate{v}<0$}
Let $c_d$ such that $F(c_d, d)\leq 1/4$. Notice that $c_d\to \infty$ as $d\to\infty$. Suppose that $\lref\leq \alpha_2/c_d$ or equivalently that $\alpha_2/\lref \geq c_d$. Then since $\|\Delta U(x)\|\to 0$, for all
$\epsilon_1>0$, there is a $K>0$ such that for all $|x|>K$ and $\lref$ small enough
\begin{align*}
2\frac{\widetilde{\mathcal{L}}V(x,v)}{V(x,v)}
&\leq -\bigg[\rate{-v} +2\lref-\frac{1}{\lref+ \rate{-v}} \langle v, \Delta U(x) v\rangle
-\lref \frac{\sqrt{\lref+\rate{-v}}}{\sqrt{\lref}}\\
&\qquad \qquad - 2\lref\frac{\sqrt{\lref+\rate{-v}}}{\sqrt{\lref}} F(\alpha_2/\lref,d)\bigg]\\
&\leq -\bigg[\rate{-v} +2\lref-\frac{1}{\lref+ \rate{-v}} \langle v, \Delta U(x) v\rangle
-\lref \frac{\sqrt{\lref+\rate{-v}}}{\sqrt{\lref}}\\
&\qquad \qquad - 2\lref\frac{\sqrt{\lref+\rate{-v}}}{\sqrt{\lref}} F(c_d,d)\bigg]\\
&\leq -\bigg[\rate{-v} +2\lref-\frac{\epsilon_1}{\lref+ \rate{-v}}
-\lref \frac{\sqrt{\lref+\rate{-v}}}{\sqrt{\lref}}\\
&\qquad \qquad - \frac{1}{2}\sqrt{\lref}\sqrt{\lref+\rate{-v}}\bigg].
\end{align*}
Let $w=\rate{-v}>0$ and consider
$$g(w) := w +2\lref - \frac{\epsilon_1}{\lref+w}-\lref\frac{\sqrt{\lref+w}}{\sqrt{\lref}}-\frac{1}{2}\sqrt{\lref}\sqrt{\lref+w}.$$
Then we obtain
$$g'(w)=1+\frac{\epsilon_1 }{(\lref+w)^2}-\frac{3 \sqrt{\lref}}{4 \sqrt{\lref+w}}.
\geq 0.$$
Thus, we have 
$$g(w)\geq g(0)= \frac{\lref}{2}-\frac{\epsilon_1}{\lref},$$
which is strictly positive as long as $\epsilon_1 < \lref^2/2$, and the result follows.
\end{proof}

\subsubsection{Position dependent refreshment}
\begin{lemma}[Lyapunov function-Varying refreshment]\label{lem:drift2}
Let the refreshment rate be equal to
$$\lref(x):=\lref+\frac{|\nabla U(x)|}{\max\{1,|x|^\epsilon\}}.$$
Then the function $V$ defined in (\ref{eq:lyapunovfn}) belongs to $\mathcal{D}(\widetilde{\mathcal{L}})$.
If in addition the assumptions of Theorem~\ref{thm:supergaussian} hold,
$V$ is a Lyapunov function as it satisfies (\ref{eq:driftCondition}).
\end{lemma}
\begin{proof}
First notice that $V$ with $\lref(x)$ as defined in the statement of the Lemma also belongs to $\mathcal{D}(\widetilde{\mathcal{L}})$ from the same arguments as in the proof of Lemma~\ref{lem:drift1}.
We now prove that $V$ satisfies \eqref{eq:driftCondition}.
From the form of \eqref{eq:driftCondition} it follows that we can assume without loss of generality that $|x| > 1$, so that
$$\lref(x) = \lref+\frac{|\nabla U(x)|}{|x|^\epsilon}.$$

First we restrict our attention to the case where $\rate{v}\neq 0$, for which we 
 compute 
\begin{align*}
\frac{\partial V}{\partial x_i}
&= \frac{1}{2}V(x,v) U_{x_i}(x) -\frac{1}{2}\frac{\re^{U(x)/2}}{(\lref(x)+\langle \nabla U(x), -v \rangle_+)^{3/2}} \left[ \frac{\partial}{\partial x_i} \lref(x) + \frac{\partial}{\partial x_i}\langle \nabla U(x), -v \rangle_+\right],\\
\langle \nabla V, v\rangle
&= \frac{1}{2}V(x,v) \rate{v} \\
&\quad -\frac{1}{2}\frac{\re^{U(x)/2}}{(\lref(x)+\langle \nabla U(x), -v \rangle_+)^{3/2}} \left[ \langle \nabla\lref(x), v\rangle
- \langle v, \Delta U(x) v\rangle \mathds{1}\{\rate{-v}>0\}\right]\\
&= \frac{1}{2}V(x,v) \left\{ \rate{v} -
\frac{\langle \nabla\lref(x), v\rangle
}{\lref(x)+\langle \nabla U(x), -v \rangle_+} +\frac{ \langle v, \Delta U(x) v\rangle \mathds{1}\{\rate{-v}>0\}}{\lref(x)+\langle \nabla U(x), -v \rangle_+} \right\}.
\end{align*}
After adding the reflection and refreshment terms we get
\begin{align*}
\widetilde{\mathcal{L}}V(x,v)
&= \frac{1}{2}V(x,v) \left\{ \rate{v} -
\frac{\langle \nabla\lref(x), v\rangle
}{\lref(x)+\langle \nabla U(x), -v \rangle_+} +\frac{ \langle v, \Delta U(x) v\rangle \mathds{1}\{\rate{-v}>0\}}{\lref(x)+\langle \nabla U(x), -v \rangle_+} \right\}\\
&+ \lref(x) \int V(x,w) \psi(\rd w) - \lref(x) V(x,v)\\
&+ \rate{v} \mathds{1}\{\rate{v}{\geq }0\} \left[ V(x,R(x) v)- V(x,v)\right],
\end{align*}
and thus
\begin{align*}
\frac{\widetilde{\mathcal{L}}V(x,v)}{V(x,v)}
&=\frac{1}{2}
\Bigg\{ \rate{v}  -\frac{\langle \nabla\lref(x), v\rangle
}{\lref(x)+\langle \nabla U(x), -v \rangle_+} +\frac{ \langle v, \Delta U(x) v\rangle \mathds{1}\{\rate{-v}>0\}}{\lref(x)+\langle \nabla U(x), -v \rangle_+} \Bigg\}\\
&+ \lref(x) \int \left[ \frac{\sqrt{\lref(x)+\rate{-v}_+}}{\sqrt{\lref(x)+\rate{w}_+}} -1\right]  \psi(\rd w) \\
&+ \rate{v} \mathds{1}\{\rate{v}{\geq}0\}  \left[\frac{\sqrt{\lref(x)+\rate{-v}_+}}{\sqrt{\lref(x)+\rate{v}_+}}- 1\right].
\end{align*}

Thus when $\rate{v}> 0$ we have
\begin{align}\label{eq:expresion_plus}
\frac{\widetilde{\mathcal{L}}V(x,v)}{V(x,v)}
&=
\frac{1}{2} \rate{v} -\frac{1}{2}\frac{\langle \nabla\lref(x), v\rangle
}{\lref(x)} + \rate{v} \left[\frac{\sqrt{\lref(x)}}{\sqrt{\lref(x)+\rate{v}}}-1\right] \notag\\
&+ \lref(x) \int_{\rate{w}\geq 0} \left[ \frac{\sqrt{\lref(x)}}{\sqrt{\lref(x)+\rate{w}_+}} -1\right]  \psi(\rd w).
\end{align}
When $\rate{v}<0$ then
\begin{align}\label{eq:expression_minus}
\frac{\widetilde{\mathcal{L}}V(x,v)}{V(x,v)}
&=\frac{1}{2}
\Bigg\{ \rate{v} -\frac{\langle \nabla\lref(x), v\rangle -  \langle v, \Delta U(x) v\rangle }
{\lref(x)+\langle \nabla U(x), -v \rangle} \Bigg\}\notag\\
&+ \lref(x)\int \left[ \frac{\sqrt{\lref(x)+\rate{-v}_+}}{\sqrt{\lref(x)+\rate{w}_+}} -1\right]  \psi(\rd w).
\end{align}
When $\rate{v}=0$, similarly to the proof of Lemma~\ref{lem:drift1}, by considering separately the case where $\langle \Delta U(x),-v\rangle>0$ and $\langle \Delta U(x),-v\rangle\leq 0$ we find that
\begin{align*}
\lim_{t\to 0+}
\frac{\rd }{\rd t} V(x+tv,v) 
&= \lim_{t\to 0+}\frac{1}{t}\left\{ V(x+tv,v)-V(x,v) \right\}\\
&= \lim_{t\to 0+}\frac{1}{t}\left\{ \frac{\exp\left(U(x+tv)/2\right)}{\sqrt{\lref(x+tv)+\langle \nabla U(x+tv),-v\rangle_+ }}-\frac{\exp\left(U(x)/2\right)}{\sqrt{\lref(x)}} \right\}\\
&= -\frac{1}{2}\frac{\exp\left(U(x)/2\right)}{\sqrt{\lref(x)}^3} \left[\langle \nabla \lref(x),v\rangle +  \langle \Delta U(x)v,-v\rangle_+ \right]\\
&=-\frac{V(x,v)}{2}\frac{\left[\langle \nabla \lref(x),v\rangle +  \langle \Delta U(x)v,-v\rangle_+ \right]}{\lref(x)}.
\end{align*}
Thus for $\rate{v}=0$, after adding the refreshment term we have
\begin{align}\label{eq:expression_zero}
\frac{\widetilde{\mathcal{L}}V(x,v)}{V(x,v)}
&=-\frac{1}{2}
\frac{\langle \nabla\lref(x), v\rangle +  \langle \Delta U(x) v,-v\rangle_+ }{ \lref(x)} \notag\\
&\qquad + \lref(x)\int \left[ \frac{\sqrt{\lref(x)}}{\sqrt{\lref(x)+\rate{w}_+}} -1\right]  \psi(\rd w)\\
&\leq -\frac{1}{2}
\frac{\langle \nabla\lref(x), v\rangle  }{ \lref(x)} + \lref(x)\int \left[ \frac{\sqrt{\lref(x)}}{\sqrt{\lref(x)+\rate{w}_+}} -1\right]  \psi(\rd w)\notag.
\end{align}
From the definition of $\lref(x)$ and the chain rule
\begin{align*}
\nabla \lref(x)
&= |x|^{-\epsilon}\nabla |\nabla U(x)| + |\nabla U(x)| \nabla \left( |x|^{-\epsilon}\right).
\end{align*}
We first compute
\begin{align*}
\frac{\partial}{\partial x_i}|\nabla U(x)|
&=
\partiald{i}\sqrt{\sum_j \left(\frac{\partial}{\partial x_j}U\right)^2}=|\nabla U(x)|^{-1} \sum_j \partiald{j}U(x) \frac{\partial^2}{\partial x_i\partial x_j}U(x),
\end{align*}
whence it follows that
\begin{align*}
\left| \nabla \left| \nabla U \right| \right|
&= |\nabla U|^{-1} | \left\{
\sum_{i=1}^d \left[ \sum_{j=1}^d  \partiald{j}U(x) \frac{\partial^2}{\partial x_i\partial x_j}U(x) \right]^2
\right\}^{1/2}\\
&\leq |\nabla U(x)|^{-1} |\nabla U|\|\Delta U\|=\|\Delta U\|
\end{align*}
Thus we have that
\begin{align*}
\frac{|\nabla \lref(x)|}{|\lref(x)|}
&\leq \frac{\|\Delta U(x)\|/|x|^\epsilon}{|\nabla U(x)|/|x|^{\epsilon}}+ \frac{|\nabla U(x)|}{|x|^{1+\epsilon} \times |\nabla U(x)|/|x|^\epsilon}\\
&=\frac{\|\Delta U(x)\|}{|\nabla U(x)|}+\frac{1}{|x|}\to 0,
\end{align*}
where we also used the fact that $|\nabla(|x|^{-\epsilon})|=\epsilon |x|^{-1-\epsilon}$.
It therefore follows that
\begin{equation}\varlimsup_{|x| \to \infty} \frac{|\langle\nabla \lref(x),v\rangle|}{\lref(x)}=0,\label{eq:lref_lim}
\end{equation}
so that this term can be ignored for large $|x|$.
Also notice that
\begin{align*}
\int_{w:\rate{w}\geq 0}\psi(\rd w) \frac{1}{\sqrt{\lref(x)+\rate{w}_+}}
&=
 \frac{1}{|\nabla U(x)|^{1/2}} \int_{w:\rate{w}\geq 0} \frac{\psi(\rd w)}{\sqrt{\frac{\lref(x)}{|\nabla U(x)|}+\left\langle \frac{\nabla U(x)}{|\nabla U(x)|},w \right\rangle
}}\\
&=
  \frac{1}{|\nabla U(x)|^{1/2}}\int_{\theta=0}^{\pi/2}
\frac{p_{\vartheta}(\rd \theta)}{\sqrt{\frac{\lref(x)}{|\nabla U(x)|}+\cos(\theta)}},
\end{align*}
where $d$ is the dimension.
As $|x|\to \infty$, our definition of $\lref(x)$ ensures that
$$\int_{\theta=0}^{\pi/2} \frac{p_{\vartheta}(\rd \theta)}{\sqrt{\frac{\lref(x)}{|\nabla U(x)|}+\cos(\theta)}}
\to \int_{\theta=0}^{\pi/2} \frac{p_{\vartheta}(\rd \theta)}{\sqrt{\cos(\theta)}}= -\frac{3 \Gamma \left(-\frac{3}{4}\right) \Gamma \left(\frac{d}{2}\right)}{8 \sqrt{\pi } \Gamma \left(\frac{d}{2}-\frac{1}{4}\right)}=:\gamma_d>0.$$
\subsubsection*{Case $\rate{v}=0$}
Thus  when $\rate{v}=0$, for $|x|$ large we have
\begin{align*}
\frac{\widetilde{\mathcal{L}}V(x,v)}{V(x,v)}
&\sim
\lref(x)\int \left[ \frac{\sqrt{\lref(x)}}{\sqrt{\lref(x)+\rate{w}_+}} -1\right]  \psi(\rd w)\\
&\sim \lref(x) \int_{\rate{w}<0}\left[1-1\right]\psi(\rd w)+ \lref(x) \left[ \frac{\sqrt{\lref(x)}}{\sqrt{|\nabla U(x)|}}\gamma_d -\frac{1}{2}\right]\\
&\sim\lref(x) \left[ \frac{\sqrt{\lref(x)}}{\sqrt{|\nabla U(x)|}}\gamma_d -\frac{1}{2}\right],
\end{align*}
and from the definition of $\lref(x)$ it easily follows that 
\begin{align*}
\frac{\widetilde{\mathcal{L}}V(x,v)}{V(x,v)}
&\sim -\frac{1}{2} \lref(x)\to -\infty.
\end{align*}
\subsubsection*{Case $\rate{v}>0$}
For $|x|$ large we have
\begin{align*}
\frac{\widetilde{\mathcal{L}}V(x,v)}{V(x,v)}
&\sim \frac{1}{2} \rate{v} + \rate{v} \left[\frac{\sqrt{\lref(x)}}{\sqrt{\lref(x)+\rate{v}}}-1\right] \\
&\qquad +\lref(x) \left[ \frac{\sqrt{\lref(x)}}{\sqrt{|\nabla U(x)|}}\gamma_d -\frac{1}{2}\right].
\end{align*}
Using the definition of $\lref(x)$, and letting
$\rate{v} = |\nabla U(x)| \cos(\theta)$ for $\theta\in[0, \pi/2)$
we have for $\rate{v}>0$ as $|x|\to \infty$
\begin{align*}
\frac{\widetilde{\mathcal{L}}V(x,v)}{V(x,v)}
&\sim
\frac{1}{2}| \nabla U(x)|\cos(\theta) + | \nabla U(x)|\cos(\theta)
\left[\frac{\sqrt{|\nabla U(x)|/|x|^\epsilon}}{\sqrt{|\nabla U(x)|/|x|^\epsilon+| \nabla U(x)|\cos(\theta)}}-1\right] \\
&\qquad + \frac{|\nabla U(x)|}{|x|^\epsilon}\left[ \frac{\sqrt{|\nabla U(x)|/|x|^\epsilon}}{\sqrt{|\nabla U(x)|}}\gamma_d -\frac{1}{2}\right]\\
&=| \nabla U(x)|\left[
\cos(\theta) \left( \frac{1}{2}  + \frac{1}{\sqrt{1+|x|^\epsilon\cos(\theta)}}-1\right)+\frac{1}{|x|^\epsilon} \left[ \frac{\gamma_d}{|x|^{\epsilon/2}} -\frac{1}{2}\right]\right] \\
&=\frac{| \nabla U(x)|}{|x|} |x|
\left[
\cos(\theta) \left(-\frac{1}{2} + \frac{1}{\sqrt{1+|x|^\epsilon\cos(\theta)}}\right)+\frac{1}{|x|^\epsilon}
\left[ \frac{\gamma_d}{|x|^{\epsilon/2}} -\frac{1}{2}\right]\right] \\
& \leq |x|
\left[
\cos(\theta) \left( -\frac{1}{2}  + \frac{1}{\sqrt{1+|x|^\epsilon\cos(\theta)}}\right)+\frac{1}{|x|^\epsilon} \left[ \frac{\gamma_d}{|x|^{\epsilon/2}} -\frac{1}{2}\right]\right],
\end{align*}
since $|\nabla U(x)|/|x|\to \infty$
 and the quantity in  brackets is clearly negative for large enough $|x|$.
Let $u=\cos(\theta)$ and $r=|x|$.
Then observe that
we can rewrite the right hand side as
$$r^{1-\epsilon}r^\epsilon u\left(\frac{1}{\sqrt{1+r^\epsilon u}}-\frac{1}{2}\right) - \frac{r^{1-\epsilon}}{2} + O(r^{1/4})
\leq  \frac{r^{1-\epsilon}}{4}-\frac{r^{1-\epsilon}}{2} +O(r^{1-3\epsilon/2}) = -\frac{\sqrt{r}}{4}+O(r^{1-3\epsilon/2}),
$$
since for $w=r^\epsilon u>0$ it can be shown that
$$w\left( \frac{1}{\sqrt{1+w}}-\frac{1}{2}\right)\leq \frac{1}{4}.$$
Thus it follows that for $\rate{v}> 0$ we have that $\varlimsup_{|x|\to\infty} \widetilde{\mathcal{L}}V/V =-\infty$.

\medskip
\subsubsection*{Case $\rate{v}<0$} From \eqref{eq:expression_minus} and \eqref{eq:lref_lim} we have as $|x|\to \infty$
\begin{align*}
&2\varlimsup_{|x|\to\infty}\frac{\widetilde{\mathcal{L}}V(x,v)}{V(x,v)}\\
&=
\varlimsup_{|x|\to\infty} \Bigg\{\rate{v} -\frac{\langle \nabla\lref(x), v\rangle -  \langle v, \Delta U(x) v\rangle }
{\lref(x)+\langle \nabla U(x), -v \rangle} \\
&\qquad+ 2\lref(x)\int \left[ \frac{\sqrt{\lref(x)+\rate{-v}_+}}{\sqrt{\lref(x)+\rate{w}_+}} -1\right]  \psi(\rd w) \Bigg\}\\
&=
\varlimsup_{|x|\to\infty} \Bigg\{\rate{v} +\frac{ \langle v, \Delta U(x) v\rangle }
{\lref(x)+\langle \nabla U(x), -v \rangle} \\
&\qquad+ \lref(x)\left[\frac{\sqrt{\lref(x)+\rate{-v}}}{\sqrt{\lref(x)}} -1\right]
+2\lref(x)\frac{\sqrt{\lref(x)+\rate{-v}_+}}{\sqrt{|\nabla U(x)|}} \gamma_d -\lref(x) \Bigg\}\\
&=
\varlimsup_{|x|\to\infty} \Bigg\{\rate{v} + \lref(x)\left[\frac{\sqrt{\lref(x)+\rate{-v}}}{\sqrt{\lref(x)}} -1\right]\\
&\qquad
+2\lref(x)\frac{\sqrt{\lref(x)+\rate{-v}_+}}{\sqrt{|\nabla U(x)|}} \gamma_d -\lref(x) \Bigg\},
\end{align*}
since $\varlimsup_{|x|\to\infty}\|\Delta U(x)\|/\lref(x)\to 0$.
Thus letting $\theta$ be the angle between $U(x)$ and $-v$, we have
\begin{align*}
\lefteqn{2\varlimsup_{|x|\to\infty}\frac{\widetilde{\mathcal{L}}V(x,v)}{V(x,v)}}\\
&=
\varlimsup_{|x|\to\infty} \Bigg\{-|\nabla U(x)|\cos(\theta) + \frac{|\nabla U(x)|}{|x|^\epsilon}
\left[\frac{\sqrt{{|\nabla U(x)|}/{|x|^\epsilon}+|\nabla U(x)|\cos(\theta)}}{\sqrt{\frac{|\nabla U(x)|}{|x|^\epsilon}}} -1\right] \\
&\qquad +2\frac{|\nabla U(x)|}{|x|^\epsilon}\left(\frac{\sqrt{\frac{|\nabla U(x)|}{|x|^\epsilon}+|\nabla U(x)|\cos(\theta)}}{\sqrt{|\nabla U(x)|}} \gamma_d -\frac{1}{2}\right) \Bigg\}\\
&=
\varlimsup_{|x|\to\infty} \Bigg\{-|\nabla U(x)|\cos(\theta) + \frac{|\nabla U(x)|}{|x|^\epsilon}
\left[\sqrt{1+|x|^\epsilon cos(\theta)} -1\right]  +2\frac{|\nabla U(x)|}{|x|^\epsilon}
\left(\sqrt{\frac{1}{|x|^\epsilon}+\cos(\theta)} \gamma_d -\frac{1}{2}\right) \Bigg\}\\
&=
\varlimsup_{|x|\to\infty} |\nabla U(x)|\Bigg\{-\cos(\theta) + \frac{1}{|x|^\epsilon}
\left[\sqrt{1+|x|^\epsilon\cos(\theta)} -1\right]  +\frac{2}{|x|^\epsilon}
\left(\sqrt{\frac{1}{|x|^\epsilon}+\cos(\theta)} \gamma_d -\frac{1}{2}\right) \Bigg\}\\
&=
\varlimsup_{|x|\to\infty} \frac{|\nabla U(x)|}{|x|} |x|\Bigg\{-\cos(\theta) + \frac{1}{|x|^\epsilon}
\left[\sqrt{1+|x|^\epsilon\cos(\theta)} -1\right]  +\frac{2}{|x|^\epsilon}
\left(\sqrt{\frac{1}{|x|^\epsilon}+\cos(\theta)} \gamma_d -\frac{1}{2}\right) \Bigg\}\\
&\leq
\varlimsup_{|x|\to\infty}  |x|\Bigg\{-\cos(\theta) + \frac{1}{|x|^\epsilon}
\left[\sqrt{1+|x|^\epsilon\cos(\theta)} -1\right]  +\frac{2}{|x|^\epsilon}
\left(\sqrt{\frac{1}{|x|^\epsilon}+\cos(\theta)} \gamma_d -\frac{1}{2}\right) \Bigg\},
\end{align*}
since the right hand side is clearly negative for $|x|$ large enough.

For $u=\cos(\theta)\in[0,1]$ define the function
$$f(u):= -u + \frac{1}{|x|^\epsilon}\left[\sqrt{1+|x|^\epsilon u} -1\right]  +\frac{2}{|x|^\epsilon}
\left(\sqrt{\frac{1}{|x|^\epsilon}+u} \gamma_d -\frac{1}{2}\right).$$
Then
$$f'(u)=-1 + \frac{|x|^\epsilon}{2|x|^\epsilon\sqrt{1+|x|^\epsilon u}}+ \frac{1}{|x|^\epsilon\sqrt{\frac{1}{|x|^\epsilon}+u}} \gamma_d.$$
This is negative for all $u\geq 0$ for $|x|$ large enough. Therefore
$$f(u)\leq f(0)= \frac{2}{|x|^\epsilon }
\left(\sqrt{\frac{1}{|x|^\epsilon }} \gamma_d -\frac{1}{2}\right)\sim -\frac{1}{|x|^\epsilon },$$
as $|x|\to \infty$. Hence
$$\varlimsup_{|x|\to \infty} \frac{\widetilde{\mathcal{L}}V(x,v)}{V(x,v)} = -\infty,$$
and the result follows.
\end{proof}

\subsection{Proof of Theorem~\ref{thm:subgaussian_tdist}}
We will frequently use \cite[Equations~(11),(13)]{J_G_12} which we  state for the reader's convenience,
\begin{equation}
\label{eq:geyer11} \nabla h(x) = 
\begin{cases}
\displaystyle\frac{f(\vert x
\vert)
\mathds{1}_d}{\vert x \vert} + \biggl[
f^\prime\bigl(\vert x \vert\bigr) - \frac{f(\vert x \vert)}{\vert
x \vert} \biggr]
\frac{xx^T}{\vert x \vert^2}, &  x\neq 0,\\
f'(0) \mathds{1}_d,& x=0,
\end{cases}
\end{equation}
%
and
\begin{equation}
\label{eq:geyer13} \det \big(\nabla h(x) \big) = 
\begin{cases}
\displaystyle f'\big(\vert x \vert\big) \bigg( \frac{f(\vert x \vert)}{\vert
x \vert}
\bigg)^{d-1}, &\quad x\neq 0, \vspace*{2pt}
\cr
f'(0)^{d},
& \quad x= 0
\end{cases}
\end{equation}
Let
$\{Z_{h,t}=(Y_t,V_t);t\geq 0\}$ be a Markov process whose generator is given by \eqref{eq:generator} with $U$ replaced by $U_h$, and write $\{P^t_h:t\geq 0\}$ for its transition kernels.
Then letting $X_t:= h(Y_t)$ for $t\geq 0$, from \cite[Corollary~3]{B_R_58}, it follows that $\{Z_t=(X_t,V_t):t\geq 0\}$ is also a Markov process with transition kernel given by
$P^t(z,A) = P^t_h(H^{-1}(z),H^{-1}(A))$ for all $A\in \mathcal{B}(\mathcal{Z})$ where $H(x,v)=(h(x),v)$.
It is also easy to see that if $Z_{h,t}$ is $\pi_h$-invariant, then ${Z_t}$ will be $\pi$-invariant--see also the discussion in \cite[Theorem~6]{J_G_12}.

Suppose now that $\{Z_{h,t}:t\geq 0\}$ is $V_h$-uniformly ergodic for some function $V_h$, that is
$$\|P_h^t(z,\cdot) - \pi_h\|_{V_h} \leq C_h V_h(z) \rho_h^t,$$
for some $C_h>0$ and $\rho_h\in (0,1)$ with $\pi_h$ admitting the density $\bar{\pi}_h(y)\psi(v)$. 
Then
we can see that
\begin{align*}
\int f \rd P^t(z,\cdot) - \int f \rd \pi
&=\int f\circ H^{} \rd P_h^t\left(H^{-1}(z),\cdot\right) - \int f\circ H \rd \pi_h.
\end{align*}
Therefore it follows that
\begin{align*}
\sup_{|f|\leq V_h\circ H^{-1}}
\left|
\int f \rd P^t(z,\cdot) - \int f \rd \pi
\right|
&=
\sup_{|f|\leq V_h\circ H^{-1}}
\left|
\int f\circ H \rd P_h^t\left(H^{-1}(z),\cdot\right) - \int f\circ H \rd \pi_h
\right|\\
&\leq \sup_{|g|\leq V_h}
\left|
\int g \rd P_h^t\left(H^{-1}(z),\cdot\right) - \int g \rd \pi_h
\right|\\
&= \| P_h^t (H^{-1}(z), \cdot) - \pi_h\|_{V_h} \leq C_h V_h\circ H^{-1} (z) \rho_h^t,
\end{align*}
whence $Z_t=H(Z_{h,t})$ is $V_h\circ H^{-1}$-uniformly ergodic.

\begin{lemma}
Under the assumptions of Theorem~\ref{thm:subgaussian_tdist},
the potentials $U_h:\mathds{R}^d \to [0,\infty)$ defined in \eqref{eq:uhdef} satisfy Assumptions~\eqref{eq:c2abs}-\eqref{eq:growth_condition}, when $h=h^{(1)}$ or $h=h^{(2)}$.
\end{lemma}
\begin{proof}
\textit{Checking  Assumption~\eqref{eq:c2abs2}.}
Notice that from equations~\eqref{eq:fpoly}, \eqref{eq:fexp} and \eqref{eq:defn_h}, the functions $h^{(i)}$, are infinitely differentiable except perhaps for $x=0$ and $|x|=1/b$ for $i=1$, or $|x|=R$ for $i=2$.
Thus $U_h$ will satisfy Assumption~(\ref{eq:c2abs}) for $|x|$ large enough and in fact everywhere except for $|x|=0, 1/b$ for $i=1$, and $|x|=0,R$ for $i=2$.  
It remains to show that the mapping $t\mapsto \langle\nabla U_h(x+tv),v\rangle$ is locally Lipschitz  at these points. First, from the definition of $f=f^{(i)}$, it follows easily that the mapping
$t\mapsto \langle\nabla U_h(x+tv),v\rangle$ will be continuous and piecewise smooth, and thus locally Lipschitz, at $|x|=1/b$ and $|x|=R$ for $i=1,2$ respectively. To deal with the remaining case $x=0$, we next show that $t\mapsto \langle\nabla U_h(tv),v\rangle$ is in fact differentiable at $t=0$. 

Recall the decomposition of $\nabla U_h$ given in \eqref{eq:gradUh}. The first term of $\eqref{eq:gradUh}$ is given by
\begin{equation}
\nabla \log \det (\nabla h(x))
=
\begin{cases}
\left[\frac{ f''(|x|)}{f'(|x|)} +(d-1) \left(\frac{f'(|x|)}{f(|x|)}-\frac{1}{|x|} \right) \right]\frac{x}{|x|}, & x\neq 0,\\
0, &x=0,
\end{cases}
\end{equation}
whence we can compute
\begin{align*}
\frac{1}{t}\left[ \langle\nabla \log \det (\nabla h(tv)),v \rangle-\langle\nabla \log \det (\nabla h(0)),v \rangle\right]
&=\frac{1}{t}\left[ \frac{f''(t)}{f'(t)} +(d-1) \left(\frac{f'(t)}{f(t)}-\frac{1}{t} \right) \right]
\Big\langle \frac{tv}{|tv|}, v \Big\rangle\\
&=\frac{1}{t}\left[ \frac{f''(t)}{f'(t)} +(d-1) \left(\frac{f'(t)}{f(t)}-\frac{1}{t} \right) \right].
\end{align*}
In the case $f=f^{(1)}$ we have
$$f(0)=0,\quad f'(0)=\frac{b\re }{2}, \quad f''(0)=0, \quad f'''(0)=b^3 \re,$$
and thus using Taylor expansions
\begin{align*}
&\lim_{t\to 0}\frac{1}{t}\left[ \frac{f''(t)}{f'(t)} +(d-1) \left(\frac{f'(t)}{f(t)}-\frac{1}{t} \right) \right]\\
&= 0+ (d-1)\lim_{t\to 0}\frac{1}{t^2 f(t)}  [tf'(t) - f(t) ]\\
&=(d-1)\lim_{t\to 0}\frac{1}{t^3 f'(0)}  
\left( tf'(0) + t^2 f''(0)+ \frac{t^3}{2} f'''(0) - f(0)-tf'(0)-\frac{t^2}{2} f''(0)-\frac{t^3}{6} f'''(0) + o(t^3)\right)\\
&=(d-1)\lim_{t\to 0}\frac{1}{t^3 f'(0)}  
\left( \frac{t^3}{2} f'''(0)-\frac{t^3}{6} f'''(0) + o(t^3)\right)
=\frac{(d-1)}{3} \frac{f'''(0)}{f'(0)}. 
\end{align*}
In the case $f=f^{(2)}$, we have for $t>0$ small enough
\begin{align*}
&\frac{1}{t}\left[ \frac{f''(t)}{f'(t)} +(d-1) \left(\frac{f'(t)}{f(t)}-\frac{1}{t} \right) \right]= 0.
\end{align*}
Thus overall $t\mapsto \langle\nabla \log \det (\nabla h(tv)),v\rangle$ is differentiable at $t=0$ and thus locally Lipshitz. 

We now deal with the second term of \eqref{eq:gradUh}. 
From \eqref{eq:geyer11} we have 
\begin{align*}
\lefteqn{\frac{1}{t}\left[\left\langle 
 \nabla h(tv)\nabla U(h(tv)),v\right\rangle -  \left\langle\nabla h(0)\nabla U(0),v\right\rangle \right]}\\
&= \frac{1}{t}\left[\left\langle\nabla h(tv)\nabla U(h(tv)),v\right\rangle -   f'(0)\langle \nabla U(0),v\rangle\right]\\
&= \frac{1}{t}\left[
\frac{f(t )}{t}\langle \nabla U(h(tv)),v\rangle-   f'(0)\langle \nabla U(0),v\rangle\right]
+ \frac{1}{t}\biggl[
f^\prime\bigl(t\bigr) - \frac{f(t)}{t} \biggr]
\frac{\langle tv,v\rangle \langle \nabla U(h(tv)),tv\rangle}{ t ^2}\\
&= I_1+I_2. 
\end{align*}
For the first term we have
\begin{align*}
I_1
&= \frac{1}{t} \frac{f(t)-tf'(0)}{t} \langle \nabla U(h(tv)),v\rangle
+ \frac{1}{t}f'(0) \left[ \langle\nabla U(h(tv)),v\rangle-\nabla U(0),v\rangle\right].
\end{align*}
Since $U$ satisfies \eqref{eq:c2abs} and $h$ is differentiable, the second term of $I_1$ clearly converges. The first term also converges since $\nabla U$ is continuous and 
$$\frac{f(t)-tf'(0)}{t^2} = \frac{f(0)+tf'(0)+t^2f''(0)+o(t^2)-tf'(0)}{t^2}.$$
For the second term we have
\begin{align*}
I_2
&=\frac{1}{t}\biggl[
f^\prime\bigl(t\bigr) - \frac{f(t)}{t} \biggr]
\langle \nabla U(h(tv)),v\rangle\\
&=\frac{f'(t)t-f(t)}{t^2}
 \langle \nabla U(h(tv)),v\rangle\\
&=\frac{f'(0)t+f''(0)t^2+o(t^2)-f(0)-tf'(0)-\frac{t^2}{2}f''(0)}{t^2}
\langle v,v\rangle \langle \nabla U(h(tv)),v\rangle\\
&\to 0.
\end{align*}
It follows that $t\mapsto \langle 
 \nabla h(x+tv)\nabla U(h(x+tv)),v\rangle$ is differentiable at $t=0$.
 
%
%
\textit{Checking Assumption~\eqref{eq:integrability}.}
For both $h = h^{(1)}$ and $h = h^{(2)}$, a change of variable leads to
\begin{align}
\int \pi_h(y) |\nabla U_h(y)| \rd y &= \int \pi_h(h^{-1}(x)) | \nabla h^{-1}(x) | |\nabla U_h(h^{-1}(x))| \rd x \nonumber\\
&= \int \pi(h(h^{-1}(x))) |\nabla U_h(h^{-1}(x))| |\nabla h(h^{-1}(x)) | |\nabla h^{-1}(x)| \rd x \nonumber \\
&= \int \pi(x) |\nabla U_h(h^{-1}(x))|  \rd x \\
&\le \int \pi(x) [ | \nabla\{h\}(h^{-1}(x)) \nabla U(x) | + |\nabla \log \det (\nabla\{h\}(h^{-1}(x)))| ] \rd x \nonumber \\
&\le \int \pi(x) [ \| \nabla\{h\}(h^{-1}(x)) \| | \nabla U(x) |  + |\nabla \log \det (\nabla\{h\}(h^{-1}(x)))| ] \rd x. \label{eq:two-terms-for-integrability-bound}
\end{align}
Here for clarity we  use the notation  $\nabla\{\cdot\}(x)$ for the gradient of the function in the bracket evaluated at $x$ and we will similarly use $\Delta\{\cdot\}(x)$ for its Hessian. 
We begin with the first term in \eqref{eq:two-terms-for-integrability-bound}. Under the assumptions of Theorem~\ref{thm:subgaussian_tdist_i} we have, for $|x| > R$ and some constant $C>0$, that $|\nabla U(x)|\leq C|x|^{-1}$ and thus
\begin{align*}
\int \pi(x)  | \nabla U(x) | \| \nabla\{h\}(h^{-1}(x)) \| \rd x
&\le C + C \int_{|x|>R} \pi(x) \frac{1}{|x|} f(|h^{-1}(x)|) \rd x \\
&\leq C + C \int_{|x|>R} \pi(x) \frac{1}{|x|} |x| \rd x \leq 2C,
\end{align*}
since clearly $f(|h^{-1}(x)|)=|x|$.

Under the assumptions of Theorem~\ref{thm:subgaussian_tdist_ii}, 
 by Assumption~\ref{assumption:directionality2}, we can assume that there exists $K>0$ such that if $|x|>K$ then $\langle x, \nabla U(x)\rangle \geq C|x|^\beta$ for some $C>0$. Thus for $|x|$ large enough, say $K/|x|<1/2$,  we have
\begin{align}
U(x)
&= U\left(K \frac{x}{|x|}\right) + \int_{t=K/|x|}^1 \frac{\rd U(tx)}{\rd t} \rd t\notag\\
&\geq  U\left(K \frac{x}{|x|}\right) + \int_{t=1/2}^1 \frac{1}{t}\langle \nabla U(tx), tx\rangle \rd t\notag\\
&\geq  U\left(K \frac{x}{|x|}\right) + C\int_{t=1/2}^1 \frac{1}{t} t^\beta |x|^\beta  \rd t \geq C|x|^\beta,\label{eq:lowerboundonU}
\end{align}
since $U\geq 0$. Therefore
\begin{align*}
\int \pi(x)  | \nabla U(x) | \| \nabla\{h\}(h^{-1}(x)) \| \rd x &\le C \Big[1+ \int_{|x|>R} \pi(x) |x|^{\beta - 1} f(|h^{-1}(x)|) \rd x \Big]\\
&\leq C\Big[ 1+ \int_{|x|>R} \re^{-C|x|^\beta} |x|^{\beta} \rd x\Big]<\infty.
\end{align*}

For the second term of \eqref{eq:two-terms-for-integrability-bound}, 
let
$$L'(x) := \left| \nabla \log \det (\nabla h(x))\right|.$$ 
From \eqref{eq:geyer13} it follows easily that 
 $L'$ is bounded for both $h = h^{(1)}$ and $h = h^{(2)}$, and thus
\begin{align*}
\int \pi(x)  |\nabla \log \det (\nabla\{h\}(h^{-1}(x)))|  \rd x &<\infty.
\end{align*}
\noindent\textit{Checking Assumption~\eqref{eq:growth_condition}.} For $h=h^{(1)}$, notice that by \cite[Lemma~4]{J_G_12}, and the fact that $h(\cdot)$ is isotropic in the sense of \cite{J_G_12}, it follows that
$$\varlimsup_{|y|\to \infty} \left| \nabla \log \det (\nabla h(y))\right| < C,$$
for some $C>0$. 
Therefore
\begin{align*}
|\nabla U_h(y)|
&\leq  \left| \nabla h(y) \nabla U(h(y))\right| + \left|\nabla \log \det (\nabla h(y))\right|\\
&\leq  \left\| \nabla h(y)\right\| \left|\nabla U(h(y))\right| + C
\intertext{and using Assumption~\ref{assumption:tail} and \eqref{eq:geyer11}}
&\leq  C\left\| \nabla h(y)\right\|/|h(y)| + C\leq C,
\end{align*}
since $\|\nabla h(y)\|\leq C |h(y)|$. Thus it follows that
\begin{equation*}
\frac{\re^{U_h(y)/2}}{\sqrt{ |\nabla U_h(y)|}} \geq C\re^{U_h(y)/2} \to \infty,
\end{equation*}
as $|y|\to \infty$ since $\re^{-U_h(y)}$ is integrable.

On the other for $h=h^{(2)}$ notice that by \cite[Lemma~2]{J_G_12}, and the fact that $h(\cdot)$ is isotropic in the sense of \cite{J_G_12}, we obtain
\begin{equation}
\varlimsup_{|y|\to \infty} \left| \nabla \log \det (\nabla h(y))\right| =0.
\label{eq:grad_log_det_h2_vanishes}
\end{equation}
Therefore
\begin{align*}
|\nabla U_h(y)|
&\leq  \left| \nabla h(y) \nabla U(h(y))\right| + \left|\nabla \log \det (\nabla h(y))\right|\\
&\leq  \left\| \nabla h(y)\right\| \left|\nabla U(h(y))\right| + C.
\end{align*}
From \eqref{eq:geyer11} it follows that $\|\nabla h(y)\|\leq C|y|^{p-1}$. Therefore, using Assumption~\ref{assumption:tail2}  
\begin{align*}
|\nabla U_h(y)|
&\leq  C\left\| \nabla h(y)\right\| |h(y)|^{\beta-1} + C\\
&\leq  C|y|^{p-1+p\beta - p}= C|y|^{p\beta - 1}.
\end{align*}
Thus
\begin{equation*}
{\frac{\re^{U_h(y)/2}}{\sqrt{ |\nabla U_h(y)|}}}
\geq C \frac{\re^{U(h(y))/2}}{\det\left( \nabla h(y)\right)\sqrt{|y|^{p\beta-1}}}.
\end{equation*}
Finally, recalling \eqref{eq:lowerboundonU}, for $|x|$ large enough, say $K/|x|<1/2$,  we have
$U(x)\geq C|x|^\beta$.
Since by definition $h(y) \sim |y|^p$, and from \eqref{eq:geyer13} $\det(\nabla h (y))$ grows at most polynomially, we obtain 
\begin{equation*}
\frac{\re^{U(h(y))/2}}{\det\left( \nabla h(y)\right)\sqrt{|y|^{p\beta-1}}} \geq \frac{\re^{C|h(y)|^\beta /2}}{\det\left( \nabla h(y)\right)\sqrt{|y|^{p\beta-1}}} 
\geq \frac{\re^{C|y|^{p\beta} /2}}{\det\left( \nabla h(y)\right)\sqrt{|y|^{p\beta-1}}} 
\to \infty.\qedhere
\end{equation*}
\end{proof}

\begin{proof}[Proof of Theorem~{\ref{thm:subgaussian_tdist_i}}]	
For notational simplicity, we assume $b = 1$ but the argument can be generalized to other values. We start by establishing the first condition of Theorem~\ref{thm:subgaussian_B}, i.e. that $U_h$ satisfies our definition of exponential tail behaviour. In the remaining, assume $|x| > b^{-1} = 1$.

By Assumption~\ref{assumption:tail} and Cauchy-Schwartz, we have for $|x|$ large enough
\[
\left| \frac{\langle x, \nabla U(x) \rangle}{|x|} \right| \le |\nabla U(x)| \le \frac{c_1}{|x|},
\]
hence $\pi$ is a \emph{sub-exponentially light density} as defined in \cite[p. 3052]{J_G_12}. This combined with Assumption~\ref{assumption:directionality}, which is equivalent to \cite[Eq. (17)]{J_G_12}, means that we can apply \cite[Theorem 3]{J_G_12} to obtain that $\pi_h$ is an \emph{exponentially light density} as defined in \cite[p. 3052]{J_G_12}. Namely there is a negative constant $c_0 < 0$ such that
\[
\varlimsup_{|x|\to \infty} \frac{- \langle x, \nabla U_h(x)\rangle}{|x|} = c_0 < 0.
\]
Applying Cauchy-Schwartz again, we obtain
\[
0 < -c_0 = \varliminf_{|x|\to \infty} \frac{\langle x, \nabla U_h(x)\rangle}{|x|} \le \varliminf_{|x|\to \infty} |\nabla U_h(x) |,
\]
which establishes the first condition of Theorem~\ref{thm:subgaussian_B}. 

We now turn our attention to the Hessian condition of Theorem~\ref{thm:subgaussian_B}. We first decompose the norm of the Hessian as follows:
\begin{equation}\label{eq:hessiandecomp}
\| \Delta U_h(x) \| \le \| \Delta \{U\circ h\}(x) \| + \| \Delta\{ \log\det(\nabla h(x))\}(x) \|.
\end{equation}
From \cite[Lemma 1]{J_G_12}, we have for $|x|\geq 1$
\[
\log\det(\nabla h(x)) = |x| + (d-1)[ \log(e^{|x|} - e/3) - \log(|x|)] =: L(|x|),
\]
hence 

\[
\Delta\{ \log\det(\nabla h(x))\}(x) = \frac{L''(|x|)}{|x|^2} x x^T + \frac{L'(|x|)}{|x|} I_d - \frac{L'(|x|)}{|x|^3} x x^T.
\]
We have
\begin{align*}
L'(r) &=  1 + (d-1) \left[\frac{e^r}{e^r - e/3} - \frac{1}{r}\right],\\
L''(r) &= (1-d) \left[\frac{e^{r+1}/3}{(e^r - e/3)^2} + \frac{1}{r^2}\right],
\end{align*}
so $|L'(r)| \to d$, $|L''(r)| \to 0$ and therefore
\begin{align*}
\varlimsup_{|x|\to \infty} \| \Delta U_h(x) \|  &\le \varlimsup_{|x|\to \infty} \| \Delta \{U\circ h\}(x) \| + \varlimsup_{|x|\to \infty} \| \Delta\{ \log\det(\nabla h(x))\}(x) \| = \varlimsup_{|x|\to \infty} \|\Delta \{U\circ h\}(x) \|.
\end{align*}

To control this remaining term, we bound the operator norm with the Frobenius norm and write
\begin{align*}
\varlimsup_{|x|\to \infty} \| \Delta \{U\circ h\}(x) \|^2 &\le \varlimsup_{|x|\to \infty} \sum_{i=1}^d \sum_{j=1}^d |\partial_i \partial_j \{U\circ h\}(x) |^2 \\
&= \varlimsup_{|x|\to \infty} \sum_{i=1}^d \sum_{j=1}^d \left| \partial_i \left\{ \sum_{k=1}^d \partial_k\{U\}(h(x)) \partial_j\{h_k\}(x) \right\} \right|^2,
\end{align*}
where we write $\partial_i\{\cdot\}(x)$ as a shorthand for the $i$-th partial derivative,  $(\nabla\{\cdot\}(x))_i$.

It is enough to bound the $d^2$ expressions of the form
\begin{align}\label{eq:two-terms}
\left| \partial_i \left\{ \sum_{k=1}^d \partial_k\{U\}(h(x))  \partial_j\{h_k\}(x)  \right\} \right| 
&\leq \sum_{k=1}^d \Big[|  \partial_i \{ \partial_k\{U\}\circ h \}(x) \}  \partial_j\{h_k\}(x)  | \\
& \qquad + |\partial_k\{U\}(h(x)) \partial_i\{  \partial_j\{h_k\}(x) \} | \} \Big].\nonumber
\end{align}
The first term in Equation~(\ref{eq:two-terms}) is controlled as follows:
\begin{align}\label{eq:ddterms}
| \partial_i \{ \partial_k\{U\}\circ h \}(x) \partial_j\{h_k\}(x) | \le | \partial_j\{h_k\}(x) | \sum_{m=1}^d | \partial_m \partial_k\{U\}(h(x))| |\partial_i\{h_m\}(x) |.
\end{align}
Using again \cite[Lemma 1]{J_G_12}, and the fact that  $|h(x)| = f(|x|) \le f'(|x|)$, for $|x|$ large enough,  
\begin{align}
|\partial_i\{h_j\}(x)| &= \left| \frac{f(|x|)}{|x|} {\mathbf 1}[i=j] + \left[ f'(|x|) - \frac{f(|x|)}{|x|} \right] \frac{x_i x_j}{|x|^2} \right| \le 3 f(|x|)\label{eq:boundgradh},
\end{align}
hence using Assumption~\ref{assumption:curvature}, for $|x|$ large enough,  
\begin{align*}
| \partial_i \{ \partial_k\{U\}\circ h \}(x) \partial_j\{h_k\}(x)  | \le d \frac{c_2}{(h(|x|))^2} (3 f(|x|))^2.
\end{align*}

The second term in Equation~(\ref{eq:two-terms}) is controlled similarly, this time using Assumption~\ref{assumption:tail}, for $|x|$ large enough,  
\begin{align}\label{eq:ddterms2}
|\partial_k\{U\}(h(x)) \partial_i\{  \partial_j\{h_k\}(x) \} | \le \frac{c_1}{h(|x|)} (8f(|x|)),
\end{align}
since it follows from \cite[Lemma 1]{J_G_12} that
\begin{equation*}
|\partial_i \partial_j \{h_k\}(x)| \le 8 f(|x|).\qedhere
\end{equation*}
\end{proof}

\begin{proof}[Proof of Theorem~\ref{thm:subgaussian_tdist_ii}]
Let $f:=f^{(2)}$, $h:=h^{(2)}$ given in \eqref{eq:fpoly} and \eqref{eq:defn_h} respectively. 
We need to check that the assumptions of Theorem~\ref{thm:supergaussian} are satisfied.
First we check that
$$\varliminf_{|x|\to\infty} \frac{ |\nabla U_h(x)|}{|x|} = \infty.$$
From \eqref{eq:gradUh}
and \eqref{eq:grad_log_det_h2_vanishes} it follows that 
\begin{align*}
\varliminf_{|x|\to\infty} \frac{|\nabla U_h (x)|}{|x|}
&= \varliminf_{|x|\to\infty} \frac{\left|\nabla h(x) \nabla U \left( h(x)\right) \right|}{|x|}.
\end{align*}
Recall from \cite[Lemma~1]{J_G_12} that for $x\neq 0$
\begin{equation}\label{eq:gradh}
\nabla h(x)= \frac{f(|x|)}{|x|} \mathds{1}_d + \bigg[ f'(|x|) - \frac{f(|x|)}{|x|} \bigg] \frac{x x^T}{|x|^2},
\end{equation}
where $\mathds{1}_d$ is the $d\times d$-identity matrix.
Therefore we have
\begin{align*}
\nabla h(x)\nabla U \left( h(x)\right)
&=\frac{f(|x|)}{|x|}\nabla U \left( h(x)\right)
+ \bigg[ f'(|x|) - \frac{f(|x|)}{|x|} \bigg] \left\langle \nabla U \left( h(x)\right), \frac{x}{|x|}\right\rangle
	\frac{x}{|x|}\\
&=f'(|x|)  \left\langle \nabla U \left( h(x)\right), \frac{x}{|x|}\right\rangle
	\frac{x}{|x|}
	+ \frac{f(|x|)}{|x|} P_x^{\perp} \nabla U \left( h(x)\right),
\end{align*}
where $P_x^{\perp}$ denotes the orthogonal projection on the plane normal to $x$.
Therefore, since by definition $h(x):=f(|x|)x/|x|$,  we have that
\begin{align*}
| \nabla h(x)\nabla U \left( h(x)\right)|
&\geq f'(|x|)  \left|\left\langle \nabla U \left( h(x)\right), \frac{x}{|x|}\right\rangle \right|\\
&= \frac{f'(|x|)}{f(|x|)}  \left|\left\langle \nabla U \left( h(x)\right), h(x)\right\rangle \right|\\
&= \frac{f'(|x|)}{f(|x|)} |h(x)|^\beta \left[ |h(x)|^{-\beta} \left|\left\langle \nabla U \left( h(x)\right), h(x)\right\rangle \right|\right].
\end{align*}
Since $|h(x)|\to \infty$ as $|x|\to \infty$, Assumption~\ref{assumption:directionality2} and the definitions of $f$ and $h$ yield
\begin{align}
\varliminf_{|x|\to\infty} \frac{|\nabla U_h (x)|}{|x|}
&\geq \varliminf_{|x|\to\infty} |x|^{-1} \left\{\frac{f'(|x|)}{f(|x|)} |h(x)|^\beta \left[ |h(x)|^{-\beta} \left|\left\langle \nabla U \left( h(x)\right), h(x)\right\rangle \right|\right] \right\}\notag\\
&\geq C \varliminf_{|x|\to\infty} |x|^{-1-1+\beta p}=C|x|^{\beta p -2} = \infty,\label{eq:lowerboundgradUh}
\end{align}
since $\beta p >2$.

Finally we need to check, that for some $\epsilon>0$ we have
$$\lim_{|x|\to\infty} \frac{\| \Delta U_h (x)\|}{|\nabla U_h (x)|}|x|^{\epsilon} = 0.$$
Recall the expression \eqref{eq:hessiandecomp}. It follows easily from the definitions of $h$, $f$ and \cite[Lemma~1, Eq.(13)]{J_G_12} that
$$\lim_{|x|\to\infty} \|\Delta \log \det (\nabla h) (x)\| =0.$$
Therefore we focus on the first term of \eqref{eq:hessiandecomp}. As in the proof of the first part of the Theorem, we need essentially to control terms of the form \eqref{eq:ddterms} and terms of the form \eqref{eq:ddterms2}. To this end, using Assumption~\ref{assumption:curvature2}, we estimate
\begin{align*}
| \partial_i \{ \partial_k\{U\} \circ h \}(x) \partial_j\{h_k\}(x) |
&\leq |\partial_j\{h_k\}(x)| \sum_{m=1}^d | \partial_m \partial_k\{U\}(h(x))| |\partial_i\{h_m\}(x)|\\
&\leq C |x|^{2p-2} |h(x)|^{\beta-2}  \leq C|x|^{2p-2+p\beta-2p} = C|x|^{p\beta-2},
\end{align*}
since from \eqref{eq:geyer11} and the definitions of $f$ and $h$ one can easily show that $|\partial_i h_k(x)|\leq |x|^{p-1}$. 
On the other hand, from Assumption~\ref{assumption:tail2} and the fact that $|\partial_i \partial_j\{h_k\}(x)|\leq C|x|^{p-2}$, which follows again from \eqref{eq:geyer11}, the remaining terms can be estimated through
\begin{align*}
|\partial_k\{U\}(h(x)) \partial_i\{ \partial_j\{h_k\}\}(x)|
&\leq |h(x)|^{\beta-1} |x|^{p-2} \leq C |x|^{p\beta -2}.
\end{align*}
Therefore combining the above with the arguments leading to \eqref{eq:lowerboundgradUh} we have that as $|x|\to \infty$
\begin{equation*}
\frac{\|\Delta U_h(x)\|}{|\nabla U_h(x)|}{|x|^\epsilon}
\leq C\frac{|x|^{\beta p -2}}{|x|^{\beta p -1} }|x|^{\epsilon}\to 0,\qedhere
\end{equation*}
\end{proof}

\begin{proof}[Proof of Theorem~\ref{thm:clt}]
Notice that if $V$ satisfies \eqref{eq:driftCondition} then for any $\varepsilon\in(0,1)$, by Jensen's inequality it follows that $\rE^z\left[ V^{1-\varepsilon}(Z_t)\right]\leq \rE^z\left[ V(Z_t)\right]^{1-\varepsilon}$. Since $\rE^z\left[ V^{\varepsilon}(Z_0)\right]=V(z)^{\varepsilon}$, it follows that
\begin{align*}
\mathcal{L} V^{1-\varepsilon}(z)
&= \frac{\rd}{\rd t} \rE^z\left[ V^{1-\varepsilon}(Z_t)\right] \bigg|_{t=0}\\
&\leq  \frac{\rd}{\rd t} \rE^z\left[ V(Z_t)\right]^{1-\varepsilon} \bigg|_{t=0}\\
&=  (1-\varepsilon)  \frac{1}{\rE^z[V(Z_t)]^\varepsilon} \frac{\rd}{\rd t}\rE^z\left[ V(Z_t)\right]\bigg|_{t=0}\\
&=(1-\varepsilon) \frac{\mathcal{L}V(z)}{V(z)^{\varepsilon}}\\
&\leq -(1-\varepsilon)\delta \frac{V(z)}{V(z)^\varepsilon} + \frac{b \mathds{1}_C(z)}{V(z)^\varepsilon},
\end{align*}
and thus $W(z):=V(z)^{1-\varepsilon}$ also satisfies \eqref{eq:driftCondition}. The result now follows from \cite[Theorem~4.3]{G_M_96}.
\end{proof}

\section*{Acknowledgements}
The authors are grateful to Fran\c{c}ois Dufour for useful discussions and pointing them towards \cite{C_90} and to Pierre Del Moral for having brought their attention to reference \cite{M_H_12}.


\end{document}